\documentclass[american,aps,prx,reprint,superscriptaddress, longbibliography,floatfix]{revtex4-2}
\usepackage[T1]{fontenc}
\usepackage{inputenc}
\usepackage{color}
\usepackage{babel}
\usepackage{physics}
\usepackage{mathtools}
\usepackage{bbding}
\usepackage{pifont}
\usepackage{textcomp}
\usepackage{wasysym}
\usepackage{amsthm}
\usepackage{nicefrac}
\usepackage{wasysym}
\usepackage{dsfont}
\usepackage{amsmath}
\usepackage{bbm}
\usepackage{amssymb}
\usepackage{bbold}
\usepackage{graphicx}
\usepackage{enumitem}
\definecolor{darkblue}{rgb}{0.0, 0.0, 0.55}

\newtheorem{theorem}{Theorem}
\newtheorem*{theorem*}{Theorem}

\newtheorem{fact}{Fact}

\newtheorem*{example*}{Example}%
\newtheorem*{remark*}{Remark}%
\newtheorem{definition}{Definition}%

\def\tr{\operatorname{tr}}
\def\supp{\operatorname{supp}}
\def\St{\operatorname{St}}
\def\Ch{\operatorname{Ch}}
\def\Pos{\operatorname{Pos}}
\def\id{\operatorname{id}}

\usepackage[colorlinks=true,
            allcolors=blue]{hyperref}

\hypersetup{
     colorlinks   = true,
     linkcolor    = blue,
     citecolor    = red,
     urlcolor     = blue,
     }

\begin{document}

\title{Maximum entropy principle for quantum processes}

\author{Siddhartha Das}
\email{das.seed@iiit.ac.in}
\affiliation{q4i, Centre for Quantum Science and Technology, International Institute of Information Technology Hyderabad, Gachibowli 500032, Telangana, India}

\author{Ujjwal Sen}\email{ujjwal@hri.res.in}
\affiliation{Harish-Chandra Research Institute, A CI of Homi Bhabha National Institute, Chhatnag Road, Jhunsi 211019, Uttar Pradesh, India}

\begin{abstract}
The maximum entropy principle, as applied to quantum systems,  is a fundamental prescript positing that for a quantum system for which we only have partial knowledge, the maximum entropy state consistent with the partial knowledge is a valuable choice as the system's state. An intriguing result is that in case the only prior knowledge is of a fixed  energy, the maximum entropy state turns out to be the thermal state, a ubiquitous state in several arenas, especially in  statistical mechanics. We extend the consequences of this principle from static quantum states to dynamic quantum processes. We establish that a quantum channel attains maximal output entropy under a fixed energy constraint if and only if it is an absolutely thermalizing channel, where the fixed output is the thermal state corresponding to that energy. Our results have potential implications for understanding the informational and thermodynamic utility of quantum channels under physical constraints. As an application, we examine the consequences for private randomness distillation from fixed energy constrained quantum processes.
\end{abstract}

\maketitle
\textit{Introduction}.--- The informational aspects of quantum systems and processes provide essential insights into both foundational principles of quantum theory and architect of emerging technologies~\cite{Ben03,HOW05,MBD+17,Das19,DBWH21}. Their significance is further underscored by the role of information theory in linking quantum mechanics with thermodynamics, enabling the study of energetics of quantum systems through thermo-informational frameworks~\cite{PSW06,BHO+13,CGG+16,Auf22}. A central concept in these frameworks is entropy, which quantifies the degree of randomness or uncertainty inherent in a system~\cite{DGP24}. The maximum entropy principle, introduced by Jaynes~\cite{Jay57a,Jay57b}, has emerged as a widely used tool in the study of thermo-informational aspects of quantum systems~\cite{GE16,BRL+19,AABS19,MSE+24}. It serves as a foundational guideline for selecting the most unbiased probability distributions consistent with known constraints, making it a natural choice for analyzing quantum systems where often the information is incomplete~\cite{CRP90,AR99,Zim08,GXLK21,BGJ23,SAVA25}.

A core and long-standing question in quantum physics is why and how thermalization emerges in quantum systems (see e.g.,~\cite{RDO08,BCH11,Ben17,MIKU18,MPS18,PB21,DMW25}), despite the fundamentally reversible nature of quantum theory due to unitary evolution of closed systems. Given the physical assumption that the  energy of the observable part of the universe is fixed by nature, and in the absence of any additional prior knowledge about the quantum process under observation, what is the appropriate way to describe the behavior of such a process? Quantum processes are the source of information carriers as well as transformers. They are formally described by completely positive, trace-preserving linear maps, also called quantum channels. To this end, we provide a general extension, or channelized analogue, of the well-known assertion of the maximum entropy principle for quantum states~\cite{Jay57a,PB21} (along with Fact~\ref{fact:max-ent-states}).

We establish the following key result, demonstrating that the absolutely thermalizing process emerges naturally as a consequence of the maximum entropy principle. Here, we define the absolutely thermalizing channel as one that maps any arbitrary input state to a unique, fixed thermal state.
\begin{theorem*}[Informal statement for Theorem~\ref{thm:max-ent}]
 Among all channels constrained to have a fixed maximum channel output energy $E$, the channel entropy is maximized if and only if the channel is absolutely thermalizing, with its fixed output being a thermal state of energy $E$.
\end{theorem*}
Quantum processes within the observable universe appear to tend toward absolute thermalization. That is, for any state preparation of a non-isolated quantum system, the system eventually thermalizes. Given that the only initial information known about these quantum processes is their fixed energy, the maximum entropy principle provides an alternative explanation for absolute thermalization. Regardless of how a quantum system is initially prepared, it evolves toward a thermal state, which is best described by processes that maximize entropy subject to the constraint of fixed energy.

 The entropy of a quantum channel captures the minimum uncertainty in its output when conditioned on a reference system~\cite{DJKR06,GW21}. It determines the thermodynamic work cost of erasing the logical output of a bipartite unitary operation, especially when ancillary outputs are accessible to the eraser~\cite{BSC+25}. Given that physical transformations of quantum states are describable as quantum channel, the thermo-informational aspects of quantum channels determine the energetics and informational utility of quantum processors and devices~\cite{FBB21,CYR21,BD26}. Quantum channels encompass quantum gates (circuits), measurements, and states (density operators)~\cite{Kra83,Sud85}. Our work can have implications for the dynamical resource theory~\cite{DBWH21,SHS24} of thermodynamics~\cite{BHO+13,BHN+15,SdS+20,BGDb} where the thermodynamic resourcefulness of a quantum channel is determined by its distinguishability from the absolutely thermalizing channel~\cite{SPSD25,BD26}.

We begin by presenting the standard notations, definitions, and foundational results as the Preliminaries, which are essential for formulating, proving, and analyzing the implications of our main result (Theorem~\ref{thm:max-ent}). We then establish our main result and discuss key consequences of the maximum entropy principle for quantum channels, and explore their implications on the private randomness capacity~\cite{YHW19} of quantum channels. Finally, we conclude by summarizing our main findings and outlining directions for future research.

\color{black}
\textit{Preliminaries}.--- We consider separable Hilbert spaces that are finite-dimensional. Let $\St(A)$ and $\Pos(A)$ denote the sets of all density operators and positive semidefinite operators, respectively, defined on the Hilbert space $\mathcal{H}_A$. The Hilbert space $\mathcal{H}_{AB}$ of a composite system $AB$ is given as $\mathcal{H}_{AB}:= \mathcal{H}_{A}\otimes\mathcal{H}_B$. We may simply denote $\mathcal{H}_A$ as $A$. The dimension of $A$ is $|A|:=\dim(A)$. $\mathbbm{1}_A$ is the identity operator on $A$. $\omega_A$ denotes operator on $A$ and for brevity we may denote it as $\omega$. For $\rho\in\St(AB)$, $\rho_B=\tr_A(\rho_{AB})$ is the reduced state on $B$. A quantum channel $\mathcal{N}_{A'\to A}$ is a completely positive, trace-preserving linear map with input on $A'$ and output on $A$. For a linear map $\mathcal{N}_{A'\to A}$ and $\psi\in\St(RA')$, $\mathcal{N}(\psi_{RA'})=\id_R\otimes\mathcal{N}_{A'\to A}(\psi_{RA'})$. Let $\Ch(A',A)$ denote the set of all quantum channels $\mathcal{N}_{A'\to A}$. Let $\mathcal{R}^{\mathbbm{1}}_{A'\to A}$ denote the uniformly mixing map, $\mathcal{R}_{A'\to A}(\rho_{A'})=\tr(\rho_{A'})\mathbbm{1}_A$ for all $\rho\in\Pos(A')$. A replacer channel $\mathcal{R}^{\omega}_{A'\to A}$ always outputs a fixed state $\omega_A$ for all possible input states $\rho_{A'}$, i.e., $\mathcal{R}^\omega_{A'\to A}(\rho_{A'})=\omega_{A}$ for all $\rho\in\St(A')$.

The Choi operator of a linear map $\mathcal{N}_{A'\to A}$ is $\Gamma^{\mathcal{N}}_{RA}=\id_R\otimes\mathcal{N}_{A'\to A}(\Gamma_{RA'})$, where $R\simeq A'$, $\id_R$ is the identity channel from $R\to R$,  $\Gamma_{RA'}:=\sum_{i,j=0}^{d-1}\ket{ii}\bra{jj}_{RA'}$ with $d=\min\{|R|,|A'|\}$. The Choi operator of a replacer channel $\mathcal{R}^{\omega}_{A'\to A}$ is a product operator $\mathbbm{1}_R\otimes\omega_A$.

The von Neumann entropy, simply referred to as entropy, of a quantum state $\rho_A$ is $S(A)_{\rho}:=S(\rho_A):=-\tr[\rho\log\rho]$. The quantum relative entropy between $\rho\in\St(A),\sigma\in\Pos(A)$ is given by $D(\rho\Vert\sigma):=\tr\left[\rho(\log\rho-\log\sigma)\right]$ if $\supp(\rho)\subseteq\supp(\sigma)$ else it is $+\infty$. For $\rho,\sigma\in\Pos(A)$, we have $D(\rho\Vert\sigma)\geq 0$ whenever $\tr[\rho]\geq \tr[\sigma]$, and $D(\rho\Vert\sigma)=0$ if and only if $\rho=\sigma$. $S(A)_{\rho}=-D(\rho_A\Vert\mathbbm{1}_A)$. For a quantum state $\rho_{AB}$, the von Neumann conditional entropy $S(A|B)_{\rho}$ quantifies uncertainty of $A$ conditioned on $B$, $S(A|B)_{\rho}=-\inf_{\sigma\in\St(B)}D(\rho_{AB}\Vert\mathbbm{1}_A\otimes\sigma_B)$, and $S(A|B)_{\rho}=S(AB)_{\rho}-S(B)_{\rho}$~\cite{CA97} (also see \cite{SB20}).

Let $\widehat{H}_A$ denote the Hamiltonian of a quantum system $A$. The  energy of $\rho\in\St(A)$ is given by $\langle \widehat{H}\rangle_{\rho}:=\tr[\widehat{H} \rho]$. For a Hamiltonian $\widehat{H}_A$, the thermal state $\gamma^{\beta}$ (also called Gibbs state) is given as $\gamma^\beta_A= {\exp(-\beta \widehat{H}_A)}/{Z^\beta_A}$~\cite{Len78}, where $\beta:=(k_BT)^{-1}$ is the inverse temperature and $Z^{\beta}_A:=\tr[\exp(-\beta \widehat{H}_A)]$. The entropy of a thermal state $\gamma^\beta$ is $S(\gamma^{\beta})=\beta \langle \widehat{H}\rangle_{\gamma^\beta} + \log Z^{\beta}$.
\begin{fact}[\cite{Jay57a}]~\label{fact:max-ent-states}
Consider a quantum system $A$ with Hamiltonian $\widehat{H}_A$. The maximum entropy among all states with a given  energy is attainable if and only if the state is the thermal (Gibbs) state of given  energy,
\begin{equation}
    \max_{\rho\in\St(A):~\tr[\widehat{H}\rho]=E} S(\rho)= S(\gamma^\beta),
\end{equation}
for $ \gamma^\beta:= \frac{\exp(-\beta \widehat{H})}{\tr\left[\exp(-\beta \widehat{H})\right]},~\langle \widehat{H}\rangle_{\gamma^{\beta}}=E$.
\end{fact}

The absolutely thermalizing channel $\mathcal{T}^\beta_{A'\to A}$ is a replacer channel such that $\mathcal{T}^\beta(\cdot)=\tr(\cdot)\gamma^\beta_A$.

\textit{Relative entropy between channels}: The quantum relative entropy between two completely positive linear maps $\mathcal{N}_{A'\to A}, \mathcal{M}_{A'\to A}$ is defined as $D[\mathcal{N}\Vert\mathcal{M}]=\sup_{\psi\in\St(RA')}D(\mathcal{N}(\psi_{RA'})\Vert\mathcal{M}(\psi_{RA'}))$~\cite{CMW16}, and it suffices to consider optimization over pure states $\psi_{RA'}$ with $R\simeq A'$. The quantum relative entropy between channels finds an operational meaning in the task of channel discrimination in Stein's setting~\cite{CMW16}. The inclusion of a reference system for the channel input can, in general, improve the discrimination between the channels compared to the setting in which no reference system is allowed~(cf.~\cite{Das19,DW19}).

\textit{Entropy of a quantum channel}: The (von Neumann) entropy of an arbitrary quantum channel $\mathcal{N}_{A'\to A}$ is defined as the negative of the relative entropy between $\mathcal{N}_{A'\to A}$ and $\mathcal{R}^{\mathbbm{1}}$~\cite{GW21} (see also~\cite{DJKR06,Yua19,SPSD25})
\begin{align}
    S[\mathcal{N}]:=& -D[\mathcal{N}\Vert\mathcal{R}^{\mathbbm{1}}]\\
    =&\inf_{\psi\in\St(RA')}\left[S(\mathcal{N}(\psi_{RA'}))-S(\psi_R)\right]\label{eq:ent-amor}\\
    =&\inf_{{\psi}\in\St(RA')}S(A|R)_{\mathcal{N}(\psi)},
\end{align}
where it suffices to optimize over pure $\psi_{RA'}$ states.
For a replacer quantum channel $\mathcal{N}_{A'\to A}=\mathcal{R}^{\omega}_{A'\to A}$, its minimum output entropy $\inf_{\rho\in\St(A')}S(\mathcal{N}(\rho))$~\cite{GGL+04,DTG17} is equal to its entropy, $S[\mathcal{R}^{\omega}]=S(\omega)=\inf_{\rho\in\St(A')}S(\mathcal{N}(\rho))$. The entropy of a quantum channel quantifies the completely bounded minimum channel output entropy~\cite{DJKR06} and also finds operational meanings in channel merging~\cite{GW21} and purity distillation~\cite{YHW19,LY20,BGD25}. It satisfies desirable axiomatic properties~\cite{GW21}:
\begin{enumerate}
    \item Non-decreasing under the action of uniformity sub-preserving superchannel $\Omega$, i.e., for $\Omega(\mathcal{R}^\mathbbm{1})\leq \mathcal{R}^{\mathbbm{1}}$~\cite{SPSD25}
    \begin{equation}
        S[\Omega(\mathcal{N})]\geq S[\mathcal{N}];
    \end{equation}
    \item Additivity: $S[\mathcal{N}\otimes\mathcal{M}]=S[\mathcal{N}]+S[\mathcal{M}]$ for any two channels $\mathcal{N}_{A'\to A}$ and $\mathcal{M}_{B'\to B}$.
    \item Reduces to the entropy of the fixed output state for a replacer channel $\mathcal{R}^\omega_{A'\to A}$, $S[\mathcal{R}^\omega]=S(\omega_A)$.
    \item Maximum if and only if the channel is uniformly mixing $\mathcal{R}^{\pi}_{A'\to A}$ for $\pi_A:=\frac{1}{|A|}\mathbbm{1}_A$, i.e., $S[\mathcal{R}^{\pi}]=\log |A|$, and minimum if and only if the channel is an isometry channel.
    \item Continuity: The difference between the entropies of quantum channels approaches zero as the diamond distance between the channels approaches zero~\cite{BGD25}.
\end{enumerate}

\textit{Main Result}.--- The maximum entropy principle is a fundamental concept in quantum statistical mechanics, asserting that when our knowledge of a quantum system is incomplete, the most reasonable description of its state is the one that maximizes entropy while remaining consistent with prior knowledge of the state~\cite{Jay57a}. A crucial implication of this principle is that if our only prior knowledge about a quantum system is its fixed  energy, then the maximum entropy state is precisely the thermal state (Fact~\ref{fact:max-ent-states}). Here we give a formal definition of the  energy of a quantum channel (cf.~\cite{Win17,vL25,BGD25}) and then prove and discuss key consequences of the maximum entropy principle for quantum channels.
\begin{definition}
    The  energy of a quantum channel $\mathcal{N}_{A'\to A}$ is defined as the maximum channel output energy. Let $\widehat{H}_A$ be the bounded Hamiltonian of $A$, then 
\begin{equation}
    \langle \widehat{H}\rangle_{\mathcal{N}}:=\sup_{\rho\in\St(A')}\tr[\widehat{H}_A\mathcal{N}(\rho_{A'})].
\end{equation}
\end{definition}
We note that, for $\widehat{H}_{RA}=\mathbbm{1}_R\otimes\widehat{H}_A+\widehat{H}_R\otimes\mathbbm{1}_A$, the maximum channel output energy is equal to (cf.~\cite{BGD25})
\begin{align}
     \langle \widehat{H}\rangle_{\mathcal{N}}&:=\sup_{\rho\in\St(A')}\tr[\widehat{H}_A\mathcal{N}(\rho_{A'})]\\ &=\sup_{\rho\in\St(RA')}\left(\langle\widehat{H}\rangle_{\id\otimes\mathcal{N}}-\langle\widehat{H}\rangle_{\id}\right),\label{eq:en-amor}
\end{align}
for the identity channel $\id_{R}$ on the reference $R$. Eq.~\eqref{eq:en-amor} has expression similar to Eq.~\eqref{eq:ent-amor} and in that sense the energy of a quantum channel can be deemed as the completely bounded maximum channel output energy, where the reference is non-interacting with the channel output. The energy of a quantum channel is additive under tensor-product of quantum channels with non-interacting outputs. That is, for $\widehat{H}_{AB}=\widehat{H}_A\otimes\mathbbm{1}_B+\mathbbm{1}_A\otimes\widehat{H}_B$ and quantum channel $\mathcal{N}_{A'\to A}\otimes\mathcal{M}_{B'\to B}$, we have $\langle \widehat{H}_{AB}\rangle_{\mathcal{N}\otimes\mathcal{M}}=\langle \widehat{H}_A\rangle_{\mathcal{N}}+\langle \widehat{H}_B\rangle_{\mathcal{M}}$. For a replacer channel $\mathcal{R}^\omega_{A'\to A}$, we have $ S[\mathcal{R}^{\omega}]=S(\omega),~\langle \widehat{H}\rangle_{\mathcal{R}^{\omega}}=\langle \widehat{H}\rangle_{{\omega}}$.

Let $\mathcal{T}^{\beta(E)}_{A'\to A}$ denote an absolutely thermalizing channel $\mathcal{T}^{\gamma^{\beta}}_{A'\to A}$ of the  energy $E$, i.e., $\langle \widehat{H}\rangle_{\mathcal{T}^{\beta(E)}}=\langle \widehat{H}\rangle_{{\gamma^{\beta(E)}}}=E$. For a given output thermal state $\gamma^\beta_A$, the absolutely thermalizing channel is unique. We first state a key implication of the maximum entropy principle for quantum processes. We then prove that the quantum channel that maximizes the entropy of quantum channels with fixed maximum  output energy is absolutely thermalizing.

\textit{Maximum entropy principle for quantum processes}: The maximum entropy principle suggests that the most accurate description of an ongoing quantum process is one that maximizes entropy while remaining consistent with any prior knowledge we have about it. If we assume our only prior knowledge about quantum processes in the observable part of the universe is that they have a fixed  energy, then, according to this principle, these processes are best characterized as absolutely thermalizing channels.  This means that no matter the initial configuration, such a channel will inevitably drive any input quantum state toward a fixed output: a thermal state consistent with the process's  energy. In essence, based on the following theorem, every initial quantum state will naturally evolve to this thermal equilibrium because the underlying quantum channel inherently leads to thermalization.

\begin{theorem}\label{thm:max-ent}
    The maximum entropy among all quantum channels $\mathcal{N}_{A'\to A}$ with a fixed  energy $\langle \widehat{H}\rangle_{\mathcal{N}}=E$, is attained if, and only if, the channel is absolutely thermalizing $\mathcal{T}^{\beta(E)}_{A'\to A}$, where $\langle \widehat{H}\rangle_{\mathcal{T}^{\beta(E)}}=E$,
    \begin{equation}
        \max_{\substack{\ \mathcal{N}\in\Ch(A',A):\\ \ \langle \widehat{H}\rangle_{\mathcal{N}}=E}} S[\mathcal{N}]=S[\mathcal{T}^{\beta(E)}].
    \end{equation}
\end{theorem}
\begin{proof} 
For a quantum channel $\mathcal{N}_{A'\to A}$, we have~\cite{GW21,DJKR06}
    \begin{align}
        S[\mathcal{N}]& =\inf_{\rho\in\St(RA')}\left[S(RA)_{\mathcal{N}(\rho)}-S(R)_{\rho}\right]\label{eq:cb-ent-a}\\
        & \leq \inf_{\rho\in\St(A')}S(A)_{\mathcal{N}(\rho)}\label{eq:cb-ent-b}\\
        &\leq \sup_{\rho\in\St(A')}S(A)_{\mathcal{N}(\rho)}.\label{eq:cb-ent-c}
    \end{align}
It follows from Ineq.~\eqref{eq:cb-ent-b} that for an arbitrary channel $\mathcal{N}_{A'\to A}$, $S[\mathcal{N}]\leq S(\mathcal{N}(\rho_{A'}))$ for all $\rho\in\St{(A')}$. If $\rho^{(1)}\in\St{(A')}$ such that $\langle \widehat{H}\rangle_{\mathcal{N}}=\tr[\widehat{H}_A\mathcal{N}(\rho^{(1)}_{A'})]$, then also $S[\mathcal{N}]\leq S(\mathcal{N}(\rho^{(1)}_{A'}))$; the existence of such $\rho^{(1)}_{A'}$ uses the facts that $\St(A')$ is compact and bounded set for $|A'|<\infty$ and the energy function is continuous on the set of states. 

We note that for every quantum channel $\mathcal{N}_{A'\to A}$ with $\langle \widehat{H}\rangle_{\mathcal{N}}=E$, there exists some replacer channels $\mathcal{R}^{\omega(E)}_{A'\to A}$, where $\omega(E)$ is to denote $\langle \widehat{H}\rangle_{\omega}=E$. Also, for every quantum channel $\mathcal{N}_{A'\to A}$ with $\langle \widehat{H}\rangle_{\mathcal{N}}=E$, there exists an absolutely thermalizing channel $\mathcal{T}^{\beta(E)}_{A'\to A}$, where $\beta(E)$ is to denote $\langle \widehat{H}\rangle_{\gamma^{\beta}}=E$.

\allowdisplaybreaks

For a quantum state $\rho_{AB}$, we know that $S(A)_{\rho}\geq S(A|B)_{\rho}$; the inequality is saturated, $S(A)_{\rho}=S(A|B)_{\rho}$, if and only if $\rho_{AB}=\rho_{A}\otimes\rho_B$. This is because $S(A)_{\rho}=S(A|B)_{\rho}$ if and only if the quantum mutual information vanishes, $I(A;B)_{\rho}:=D(\rho_{AB}\Vert\rho_A\otimes\rho_B)=0$. We have
\begin{align}
  0
  &\leq \sup_{\psi\in\St(RA')} I(R;A)_{\mathcal{N}(\psi)} \\
  & = \sup_{\psi\in\St(RA')}\left[
      S(A)_{\mathcal{N}(\psi)}
      + S(R)_{\psi}
      - S(RA)_{\mathcal{N}(\psi)}
    \right] \\
  & \leq \sup_{\psi\in\St(A')} S(A)_{\mathcal{N}(\psi)}
      + \sup_{\psi\in\St(RA')}\left(-S(A|R)_{\mathcal{N}(\psi)}\right) \\
  & = \sup_{\psi\in\St(A')} S(A)_{\mathcal{N}(\psi)}
      - \inf_{\psi\in\St(RA')} S(A|R)_{\mathcal{N}(\psi)} \\
  & = \sup_{\psi\in\St(A')} S(A)_{\mathcal{N}(\psi)}
      - S[\mathcal{N}],
\end{align}
where the first inequality follows from the fact that $I(A;B)_{\rho}\geq 0$ for all quantum states $\rho_{AB}$. $S[\mathcal{N}]=\sup_{\psi\in\St(A')}S(A)_{\mathcal{N}(\psi)}$ implies that $\sup_{\psi\in\St(RA')}I(R;A)_{\mathcal{N}(\psi)}=0$, which is true if and only if the channel $\mathcal{N}$ is a replacer channel.

It holds from Ineq.~\eqref{eq:cb-ent-c} that for an arbitrary quantum channel $\mathcal{N}_{A'\to A}$, $S[\mathcal{N}]\leq  \sup_{\rho\in\St(A')}S(\mathcal{N}(\rho_{A'}))$, and the upper bound is achievable only for replacer channels $\mathcal{R}^{\sigma}_{A'\to A}$ with $\sigma_A$ such that $\sup_{\rho\in\St(A')}S(\mathcal{N}(\rho_{A'}))=S(\sigma_A)$. It then directly follows from Fact~\ref{fact:max-ent-states} that among all replacer channels $\mathcal{R}^{\omega}_{A'\to A}$ with $\langle \widehat{H}\rangle_{\mathcal{R}^{\omega}}=E$, the entropy is maximum if and only if the replacer channel is absolutely thermalizing $\mathcal{T}^{\beta}$ and $\langle \widehat{H}\rangle_{\gamma^{\beta}}=E$. This concludes the proof.
\end{proof}

The theorem statement above can be generalized to the cases where chemical potentials~\cite{HHW67,PB21} are also considered. Fact~\ref{fact:max-ent-states} will get modified by considering the generalized thermal (Gibbs) state, and Theorem~\ref{thm:max-ent} will be revised accordingly. We leave formal statement and rigorous proof for future work.

The absolutely thermalizing channel $\mathcal{T}^{\beta(E)}_{A'\to A}$ can be simulated via unitary SWAP gate and the thermal ancillary (bath) state, as (see e.g.,~\cite{BSC+25}) $  \mathcal{T}^{\beta(E)}_{A'\to A}(\cdot)=\tr_E\left[\mathcal{U}^{\rm SWAP}_{A'E'\to AE}(~\cdot\otimes\gamma^{\beta}_{E'})\right]$, 
where $\mathcal{U}^{\rm SWAP}_{A'E'\to AE}(\cdot)=U^{\rm SWAP}(\cdot)(U^{\rm SWAP})^{\dag}$ is unitary SWAP channel with $A',A$ denoting logical systems and $E',E$ denoting bath systems, $U^{\rm SWAP}_{A'E'\to AE}\ket{i}_{A'}\ket{j}_{E'}=\ket{j}_{A}\ket{i}_{E}$. Provided that the individual Hamiltonians $\widehat{H}_{A'},\widehat{H}_{E'}$ are identical, the total Hamiltonian $\widehat{H}_{A'+E'}$ is 
$\widehat{H}_{A'+E'}=\widehat{H}_{A'}\otimes\mathbbm{1}_{E'}+\mathbbm{1}_{A'}\otimes\widehat{H}_{E'}$,
and the unitary SWAP channel is energy-preserving operation, $\left[{U}^{\rm SWAP}_{A'E'\to AE},\widehat{H}_{A'+E'}\right]=0$. There are other physical methods that realize absolute thermalization, see for instances~\cite{SZS+05,RDO08,MIKU18,HMG19,MDP22}.

\textit{Private or intrinsic randomness}: There is a trade-off relation between the entropy $S[\mathcal{N}]$ of a quantum channel $\mathcal{N}_{A'\to A}$, for $|A|<\infty$, and its private randomness capacity $P_{\rm random}[\mathcal{N}]$~\cite{YHW19},~\cite[Eq.~(128)]{DGP24}
\begin{equation}\label{eq:e-p-tradeoff}
    S[\mathcal{N}]+P_{\rm random}[\mathcal{N}]=\log |A|.
\end{equation}
The private randomness capacity of a channel $\mathcal{N}_{A'\to A}$ is the maximum rate, in an asymptotic setting of \textit{i.i.d.} uses of the channel, at which the receiver accessing output $A$ can extract private randomness, against an eavesdropper accessing extension $E$ of an isometric channel extension $\mathcal{V}^{\mathcal{N}}_{A'\to AE}$ of the channel $\mathcal{N}$, when the sender sends states through the channel. Under the physically natural assumption of the energy constraint, the private randomness capacity is minimal for the absolutely thermalizing channel. The minimum private randomness capacity among all quantum channels $\mathcal{N}_{A'\to A}$ with a fixed  energy $\langle \widehat{H}\rangle_{\mathcal{N}}=E$, is attained if, and only if, the channel is absolutely thermalizing $\mathcal{T}^{\beta(E)}_{A'\to A}$,
    \begin{equation}
        \min_{\substack{\ \mathcal{N}\in\Ch(A',A):\\ \ \langle \widehat{H}\rangle_{\mathcal{N}}=E}} P_{\rm random}[\mathcal{N}]=P_{\rm random}[\mathcal{T}^{\beta(E)}].
    \end{equation}

The trade-off relation~\eqref{eq:e-p-tradeoff} resembles closely with the trade-off relation between the entropy $S(\rho_A)$ of a quantum state $\rho_A$ and its objective information $I(\rho_A):= D(\rho_A\Vert\pi_A)$~\cite{Zur01,HHO03,HHH+03,MDP22,SSG+25}, where $\pi_A:=\mathbbm{1}/|A|$ is the maximally mixed state and $|A|<\infty$,
\begin{equation}\label{eq:e-p-state}
    S(\rho)+I(\rho)=\log|A|. 
\end{equation}
The objective information of a state is related to the amount of extractable work~\cite{OHHH02,HHH+05} and intrinsic randomness~\cite{MCS+24} in the state. We can express the private randomness capacity $P_{\rm random}[\mathcal{N}]$ of a quantum channel $\mathcal{N}_{A'\to A}$ as $P_{\rm random}[\mathcal{N}]=D[\mathcal{N}\Vert\mathcal{R}^{\pi}]$, where $\mathcal{R}^{\pi}_{A'\to A}$ is the uniformly mixing channel or the completely depolarizing channel. The entropic quantity $D[\mathcal{N}\Vert\mathcal{R}^{\pi}]$ can be deemed as a measure of the intrinsic randomness of a quantum channel $\mathcal{N}_{A'\to A}$~\cite[Section~IV.A.]{DGP24} (see also~\cite{BGDb}), given that $D(\rho_A\Vert\pi_A)$ is a measure of the intrinsic randomness of a state $\rho_A$. Trade-off relations~\eqref{eq:e-p-state} and \eqref{eq:e-p-tradeoff} are state--channel analogue of each other. For an arbitrary replacer channel $\mathcal{R}^{\omega}_{A'\to A}$ we have $S[\mathcal{R}^{\omega}]=S(\omega)$ and $P_{\rm random}[\mathcal{R}^{\omega}]=I(\omega)$. We can arguably also interpret the absolute thermalization as a consequence of the \textit{minimum intrinsic randomness principle}, a quantum information-theoretic dual to the maximum entropy principle.

\textit{Discussion}.--- Understanding the thermo-informational aspects of quantum processes is crucial for both advancing quantum technologies and deepening our foundational knowledge of quantum mechanics. Information-theoretic quantities like entropy, conditional entropy, mutual information, and conditional mutual information are vital for discerning the causal structure and unique properties of these underlying quantum dynamics~\cite{DGP24,BD26}.

In this work, we have utilized the maximum entropy principle, a robust framework for drawing unbiased inferences from incomplete data, and applied it specifically to quantum channels. This provides an axiomatic approach to understanding absolute thermalization~\cite{RGE12,GE16,AABS19,MSE+24}. Our findings, which are based on the maximum entropy principle, offer new insights into the research area of quantum-to-classical transitions~\cite{PSW06,LC09,LWG21,SSK24}.

Furthermore, we have developed an argument for evaluating the thermodynamic and informational resourcefulness of a quantum channel under fixed energy constraints. This can be achieved by quantifying how much a given channel deviates from an absolutely thermalizing channel~\cite{SPSD25}. This approach holds promising implications for thermodynamic resource theories~\cite{BHO+13,SdS+20} related to quantum channels. Our main result imply that the absolutely thermal channel can be deemed free in the dynamical resource theory of athermality, thus paving path for introducing the dynamical resource theory of athermality~\cite{BGD25,BD26}.

\textit{Note added}: This topic is also studied in a parallel work by Philippe Faist and Sumeet Khatri, arXiv:2508.03993. Some exploratory suggestions in our work have been  recently studied in arXiv:2510.12790, arXiv:2510.23731, and arXiv:2604.01217, where the dynamical resource theory of athermality is introduced and its implications in quantum thermodynamics and quantum information processing are investigated. The athermality monotones of a quantum channel are defined as the generalized divergences between the channel with respect to the absolutely thermalizing channel, see~\cite[Eqs.~(67)-(70)]{SPSD25}.

\medskip

\textit{Acknowledgments}.--- The authors thank Himanshu Badhani, Karol Horodecki, and Uttam Singh for discussions. S.D. acknowledges support from the SERB (now ANRF), Department of Science and Technology, Government of India, under grant no. SRG/2023/000217 and the Ministry of Electronics and Information Technology (MeitY), Government of India, under grant no. 4(3)/2024-ITEA, and the National Science Centre, Poland, under grant Opus 25, 2023/49/B/ST2/02468. S.D. thanks Harish-Chandra Research Institute, Prayagraj (Allahabad), India and University of Gda\'nsk, Gda\'nsk, Poland for the hospitality during his visits.
\bibliography{output}

\begin{thebibliography}{70}%
\makeatletter
\providecommand \@ifxundefined [1]{%
 \@ifx{#1\undefined}
}%
\providecommand \@ifnum [1]{%
 \ifnum #1\expandafter \@firstoftwo
 \else \expandafter \@secondoftwo
 \fi
}%
\providecommand \@ifx [1]{%
 \ifx #1\expandafter \@firstoftwo
 \else \expandafter \@secondoftwo
 \fi
}%
\providecommand \natexlab [1]{#1}%
\providecommand \enquote  [1]{``#1''}%
\providecommand \bibnamefont  [1]{#1}%
\providecommand \bibfnamefont [1]{#1}%
\providecommand \citenamefont [1]{#1}%
\providecommand \href@noop [0]{\@secondoftwo}%
\providecommand \href [0]{\begingroup \@sanitize@url \@href}%
\providecommand \@href[1]{\@@startlink{#1}\@@href}%
\providecommand \@@href[1]{\endgroup#1\@@endlink}%
\providecommand \@sanitize@url [0]{\catcode `\\12\catcode `\$12\catcode `\&12\catcode `\#12\catcode `\^12\catcode `\_12\catcode `\%12\relax}%
\providecommand \@@startlink[1]{}%
\providecommand \@@endlink[0]{}%
\providecommand \url  [0]{\begingroup\@sanitize@url \@url }%
\providecommand \@url [1]{\endgroup\@href {#1}{\urlprefix }}%
\providecommand \urlprefix  [0]{URL }%
\providecommand \Eprint [0]{\href }%
\providecommand \doibase [0]{https://doi.org/}%
\providecommand \selectlanguage [0]{\@gobble}%
\providecommand \bibinfo  [0]{\@secondoftwo}%
\providecommand \bibfield  [0]{\@secondoftwo}%
\providecommand \translation [1]{[#1]}%
\providecommand \BibitemOpen [0]{}%
\providecommand \bibitemStop [0]{}%
\providecommand \bibitemNoStop [0]{.\EOS\space}%
\providecommand \EOS [0]{\spacefactor3000\relax}%
\providecommand \BibitemShut  [1]{\csname bibitem#1\endcsname}%
\let\auto@bib@innerbib\@empty
\bibitem [{\citenamefont {Bennett}(2003)}]{Ben03}%
  \BibitemOpen
  \bibfield  {author} {\bibinfo {author} {\bibfnamefont {C.~H.}\ \bibnamefont {Bennett}},\ }\bibfield  {title} {\bibinfo {title} {Notes on {L}andauer's principle, reversible computation, and {M}axwell's demon},\ }\href {https://doi.org/https://doi.org/10.1016/S1355-2198(03)00039-X} {\bibfield  {journal} {\bibinfo  {journal} {Studies in History and Philosophy of Science Part B: Studies in History and Philosophy of Modern Physics}\ }\textbf {\bibinfo {volume} {34}},\ \bibinfo {pages} {501} (\bibinfo {year} {2003})},\ \bibinfo {note} {{Q}uantum Information and Computation}\BibitemShut {NoStop}%
\bibitem [{\citenamefont {Horodecki}\ \emph {et~al.}(2005{\natexlab{a}})\citenamefont {Horodecki}, \citenamefont {Oppenheim},\ and\ \citenamefont {Winter}}]{HOW05}%
  \BibitemOpen
  \bibfield  {author} {\bibinfo {author} {\bibfnamefont {M.}~\bibnamefont {Horodecki}}, \bibinfo {author} {\bibfnamefont {J.}~\bibnamefont {Oppenheim}},\ and\ \bibinfo {author} {\bibfnamefont {A.}~\bibnamefont {Winter}},\ }\bibfield  {title} {\bibinfo {title} {Partial quantum information},\ }\href {https://doi.org/10.1038/nature03909} {\bibfield  {journal} {\bibinfo  {journal} {Nature}\ }\textbf {\bibinfo {volume} {436}},\ \bibinfo {pages} {673–676} (\bibinfo {year} {2005}{\natexlab{a}})}\BibitemShut {NoStop}%
\bibitem [{\citenamefont {Majenz}\ \emph {et~al.}(2017)\citenamefont {Majenz}, \citenamefont {Berta}, \citenamefont {Dupuis}, \citenamefont {Renner},\ and\ \citenamefont {Christandl}}]{MBD+17}%
  \BibitemOpen
  \bibfield  {author} {\bibinfo {author} {\bibfnamefont {C.}~\bibnamefont {Majenz}}, \bibinfo {author} {\bibfnamefont {M.}~\bibnamefont {Berta}}, \bibinfo {author} {\bibfnamefont {F.}~\bibnamefont {Dupuis}}, \bibinfo {author} {\bibfnamefont {R.}~\bibnamefont {Renner}},\ and\ \bibinfo {author} {\bibfnamefont {M.}~\bibnamefont {Christandl}},\ }\bibfield  {title} {\bibinfo {title} {Catalytic decoupling of quantum information},\ }\href {https://doi.org/10.1103/PhysRevLett.118.080503} {\bibfield  {journal} {\bibinfo  {journal} {Physical Review Letters}\ }\textbf {\bibinfo {volume} {118}},\ \bibinfo {pages} {080503} (\bibinfo {year} {2017})}\BibitemShut {NoStop}%
\bibitem [{\citenamefont {Das}(2019)}]{Das19}%
  \BibitemOpen
  \bibfield  {author} {\bibinfo {author} {\bibfnamefont {S.}~\bibnamefont {Das}},\ }\href {https://arxiv.org/abs/1901.05895} {\bibinfo {title} {Bipartite quantum interactions: Entangling and information processing abilities}} (\bibinfo {year} {2019}),\ \bibinfo {note} {arXiv:1901.05895}\BibitemShut {NoStop}%
\bibitem [{\citenamefont {Das}\ \emph {et~al.}(2021)\citenamefont {Das}, \citenamefont {B\"auml}, \citenamefont {Winczewski},\ and\ \citenamefont {Horodecki}}]{DBWH21}%
  \BibitemOpen
  \bibfield  {author} {\bibinfo {author} {\bibfnamefont {S.}~\bibnamefont {Das}}, \bibinfo {author} {\bibfnamefont {S.}~\bibnamefont {B\"auml}}, \bibinfo {author} {\bibfnamefont {M.}~\bibnamefont {Winczewski}},\ and\ \bibinfo {author} {\bibfnamefont {K.}~\bibnamefont {Horodecki}},\ }\bibfield  {title} {\bibinfo {title} {Universal limitations on quantum key distribution over a network},\ }\href {https://doi.org/10.1103/PhysRevX.11.041016} {\bibfield  {journal} {\bibinfo  {journal} {Physical Review X}\ }\textbf {\bibinfo {volume} {11}},\ \bibinfo {pages} {041016} (\bibinfo {year} {2021})}\BibitemShut {NoStop}%
\bibitem [{\citenamefont {Popescu}\ \emph {et~al.}(2006)\citenamefont {Popescu}, \citenamefont {Short},\ and\ \citenamefont {Winter}}]{PSW06}%
  \BibitemOpen
  \bibfield  {author} {\bibinfo {author} {\bibfnamefont {S.}~\bibnamefont {Popescu}}, \bibinfo {author} {\bibfnamefont {A.~J.}\ \bibnamefont {Short}},\ and\ \bibinfo {author} {\bibfnamefont {A.}~\bibnamefont {Winter}},\ }\bibfield  {title} {\bibinfo {title} {Entanglement and the foundations of statistical mechanics},\ }\href {https://doi.org/10.1038/nphys444} {\bibfield  {journal} {\bibinfo  {journal} {Nature Physics}\ }\textbf {\bibinfo {volume} {2}},\ \bibinfo {pages} {754–758} (\bibinfo {year} {2006})}\BibitemShut {NoStop}%
\bibitem [{\citenamefont {Brand\~ao}\ \emph {et~al.}(2013)\citenamefont {Brand\~ao}, \citenamefont {Horodecki}, \citenamefont {Oppenheim}, \citenamefont {Renes},\ and\ \citenamefont {Spekkens}}]{BHO+13}%
  \BibitemOpen
  \bibfield  {author} {\bibinfo {author} {\bibfnamefont {F.~G. S.~L.}\ \bibnamefont {Brand\~ao}}, \bibinfo {author} {\bibfnamefont {M.}~\bibnamefont {Horodecki}}, \bibinfo {author} {\bibfnamefont {J.}~\bibnamefont {Oppenheim}}, \bibinfo {author} {\bibfnamefont {J.~M.}\ \bibnamefont {Renes}},\ and\ \bibinfo {author} {\bibfnamefont {R.~W.}\ \bibnamefont {Spekkens}},\ }\bibfield  {title} {\bibinfo {title} {Resource theory of quantum states out of thermal equilibrium},\ }\href {https://doi.org/10.1103/PhysRevLett.111.250404} {\bibfield  {journal} {\bibinfo  {journal} {Physical Review Letters}\ }\textbf {\bibinfo {volume} {111}},\ \bibinfo {pages} {250404} (\bibinfo {year} {2013})}\BibitemShut {NoStop}%
\bibitem [{\citenamefont {Cabello}\ \emph {et~al.}(2016)\citenamefont {Cabello}, \citenamefont {Gu}, \citenamefont {G\"uhne}, \citenamefont {Larsson},\ and\ \citenamefont {Wiesner}}]{CGG+16}%
  \BibitemOpen
  \bibfield  {author} {\bibinfo {author} {\bibfnamefont {A.}~\bibnamefont {Cabello}}, \bibinfo {author} {\bibfnamefont {M.}~\bibnamefont {Gu}}, \bibinfo {author} {\bibfnamefont {O.}~\bibnamefont {G\"uhne}}, \bibinfo {author} {\bibfnamefont {J.-A.}\ \bibnamefont {Larsson}},\ and\ \bibinfo {author} {\bibfnamefont {K.}~\bibnamefont {Wiesner}},\ }\bibfield  {title} {\bibinfo {title} {Thermodynamical cost of some interpretations of quantum theory},\ }\href {https://doi.org/10.1103/PhysRevA.94.052127} {\bibfield  {journal} {\bibinfo  {journal} {Physical Review A}\ }\textbf {\bibinfo {volume} {94}},\ \bibinfo {pages} {052127} (\bibinfo {year} {2016})}\BibitemShut {NoStop}%
\bibitem [{\citenamefont {Auff\`eves}(2022)}]{Auf22}%
  \BibitemOpen
  \bibfield  {author} {\bibinfo {author} {\bibfnamefont {A.}~\bibnamefont {Auff\`eves}},\ }\bibfield  {title} {\bibinfo {title} {Quantum technologies need a quantum energy initiative},\ }\href {https://doi.org/10.1103/PRXQuantum.3.020101} {\bibfield  {journal} {\bibinfo  {journal} {PRX Quantum}\ }\textbf {\bibinfo {volume} {3}},\ \bibinfo {pages} {020101} (\bibinfo {year} {2022})}\BibitemShut {NoStop}%
\bibitem [{\citenamefont {Das}\ \emph {et~al.}(2024)\citenamefont {Das}, \citenamefont {Goswami},\ and\ \citenamefont {Pandey}}]{DGP24}%
  \BibitemOpen
  \bibfield  {author} {\bibinfo {author} {\bibfnamefont {S.}~\bibnamefont {Das}}, \bibinfo {author} {\bibfnamefont {K.}~\bibnamefont {Goswami}},\ and\ \bibinfo {author} {\bibfnamefont {V.}~\bibnamefont {Pandey}},\ }\href {https://arxiv.org/abs/2410.01740} {\bibinfo {title} {Conditional entropy and information of quantum processes}} (\bibinfo {year} {2024}),\ \bibinfo {note} {arXiv:2410.01740}\BibitemShut {NoStop}%
\bibitem [{\citenamefont {Jaynes}(1957{\natexlab{a}})}]{Jay57a}%
  \BibitemOpen
  \bibfield  {author} {\bibinfo {author} {\bibfnamefont {E.~T.}\ \bibnamefont {Jaynes}},\ }\bibfield  {title} {\bibinfo {title} {Information theory and statistical mechanics},\ }\href {https://doi.org/10.1103/PhysRev.106.620} {\bibfield  {journal} {\bibinfo  {journal} {Physical Review}\ }\textbf {\bibinfo {volume} {106}},\ \bibinfo {pages} {620} (\bibinfo {year} {1957}{\natexlab{a}})}\BibitemShut {NoStop}%
\bibitem [{\citenamefont {Jaynes}(1957{\natexlab{b}})}]{Jay57b}%
  \BibitemOpen
  \bibfield  {author} {\bibinfo {author} {\bibfnamefont {E.~T.}\ \bibnamefont {Jaynes}},\ }\bibfield  {title} {\bibinfo {title} {Information theory and statistical mechanics. ii},\ }\href {https://doi.org/10.1103/PhysRev.108.171} {\bibfield  {journal} {\bibinfo  {journal} {Physical Review}\ }\textbf {\bibinfo {volume} {108}},\ \bibinfo {pages} {171} (\bibinfo {year} {1957}{\natexlab{b}})}\BibitemShut {NoStop}%
\bibitem [{\citenamefont {Gogolin}\ and\ \citenamefont {Eisert}(2016)}]{GE16}%
  \BibitemOpen
  \bibfield  {author} {\bibinfo {author} {\bibfnamefont {C.}~\bibnamefont {Gogolin}}\ and\ \bibinfo {author} {\bibfnamefont {J.}~\bibnamefont {Eisert}},\ }\bibfield  {title} {\bibinfo {title} {Equilibration, thermalisation, and the emergence of statistical mechanics in closed quantum systems},\ }\href {https://doi.org/10.1088/0034-4885/79/5/056001} {\bibfield  {journal} {\bibinfo  {journal} {Reports on Progress in Physics}\ }\textbf {\bibinfo {volume} {79}},\ \bibinfo {pages} {056001} (\bibinfo {year} {2016})}\BibitemShut {NoStop}%
\bibitem [{\citenamefont {Bera}\ \emph {et~al.}(2019)\citenamefont {Bera}, \citenamefont {Riera}, \citenamefont {Lewenstein}, \citenamefont {Khanian},\ and\ \citenamefont {Winter}}]{BRL+19}%
  \BibitemOpen
  \bibfield  {author} {\bibinfo {author} {\bibfnamefont {M.~N.}\ \bibnamefont {Bera}}, \bibinfo {author} {\bibfnamefont {A.}~\bibnamefont {Riera}}, \bibinfo {author} {\bibfnamefont {M.}~\bibnamefont {Lewenstein}}, \bibinfo {author} {\bibfnamefont {Z.~B.}\ \bibnamefont {Khanian}},\ and\ \bibinfo {author} {\bibfnamefont {A.}~\bibnamefont {Winter}},\ }\bibfield  {title} {\bibinfo {title} {Thermodynamics as a consequence of information conservation},\ }\href {https://doi.org/10.22331/q-2019-02-14-121} {\bibfield  {journal} {\bibinfo  {journal} {Quantum}\ }\textbf {\bibinfo {volume} {3}},\ \bibinfo {pages} {121} (\bibinfo {year} {2019})}\BibitemShut {NoStop}%
\bibitem [{\citenamefont {Abanin}\ \emph {et~al.}(2019)\citenamefont {Abanin}, \citenamefont {Altman}, \citenamefont {Bloch},\ and\ \citenamefont {Serbyn}}]{AABS19}%
  \BibitemOpen
  \bibfield  {author} {\bibinfo {author} {\bibfnamefont {D.~A.}\ \bibnamefont {Abanin}}, \bibinfo {author} {\bibfnamefont {E.}~\bibnamefont {Altman}}, \bibinfo {author} {\bibfnamefont {I.}~\bibnamefont {Bloch}},\ and\ \bibinfo {author} {\bibfnamefont {M.}~\bibnamefont {Serbyn}},\ }\bibfield  {title} {\bibinfo {title} {Colloquium: Many-body localization, thermalization, and entanglement},\ }\href {https://doi.org/10.1103/RevModPhys.91.021001} {\bibfield  {journal} {\bibinfo  {journal} {Reviews of Modern Physics}\ }\textbf {\bibinfo {volume} {91}},\ \bibinfo {pages} {021001} (\bibinfo {year} {2019})}\BibitemShut {NoStop}%
\bibitem [{\citenamefont {Mark}\ \emph {et~al.}(2024)\citenamefont {Mark}, \citenamefont {Surace}, \citenamefont {Elben}, \citenamefont {Shaw}, \citenamefont {Choi}, \citenamefont {Refael}, \citenamefont {Endres},\ and\ \citenamefont {Choi}}]{MSE+24}%
  \BibitemOpen
  \bibfield  {author} {\bibinfo {author} {\bibfnamefont {D.~K.}\ \bibnamefont {Mark}}, \bibinfo {author} {\bibfnamefont {F.}~\bibnamefont {Surace}}, \bibinfo {author} {\bibfnamefont {A.}~\bibnamefont {Elben}}, \bibinfo {author} {\bibfnamefont {A.~L.}\ \bibnamefont {Shaw}}, \bibinfo {author} {\bibfnamefont {J.}~\bibnamefont {Choi}}, \bibinfo {author} {\bibfnamefont {G.}~\bibnamefont {Refael}}, \bibinfo {author} {\bibfnamefont {M.}~\bibnamefont {Endres}},\ and\ \bibinfo {author} {\bibfnamefont {S.}~\bibnamefont {Choi}},\ }\bibfield  {title} {\bibinfo {title} {Maximum entropy principle in deep thermalization and in {H}ilbert-space ergodicity},\ }\href {https://doi.org/10.1103/PhysRevX.14.041051} {\bibfield  {journal} {\bibinfo  {journal} {Physical Review X}\ }\textbf {\bibinfo {volume} {14}},\ \bibinfo {pages} {041051} (\bibinfo {year} {2024})}\BibitemShut {NoStop}%
\bibitem [{\citenamefont {Canosa}\ \emph {et~al.}(1990)\citenamefont {Canosa}, \citenamefont {Rossignoli},\ and\ \citenamefont {Plastino}}]{CRP90}%
  \BibitemOpen
  \bibfield  {author} {\bibinfo {author} {\bibfnamefont {N.}~\bibnamefont {Canosa}}, \bibinfo {author} {\bibfnamefont {R.}~\bibnamefont {Rossignoli}},\ and\ \bibinfo {author} {\bibfnamefont {A.}~\bibnamefont {Plastino}},\ }\bibfield  {title} {\bibinfo {title} {Maximum entropy principle for many-body ground states},\ }\href {https://doi.org/https://doi.org/10.1016/0375-9474(90)90083-X} {\bibfield  {journal} {\bibinfo  {journal} {Nuclear Physics A}\ }\textbf {\bibinfo {volume} {512}},\ \bibinfo {pages} {492} (\bibinfo {year} {1990})}\BibitemShut {NoStop}%
\bibitem [{\citenamefont {Abe}\ and\ \citenamefont {Rajagopal}(1999)}]{AR99}%
  \BibitemOpen
  \bibfield  {author} {\bibinfo {author} {\bibfnamefont {S.}~\bibnamefont {Abe}}\ and\ \bibinfo {author} {\bibfnamefont {A.~K.}\ \bibnamefont {Rajagopal}},\ }\bibfield  {title} {\bibinfo {title} {Quantum entanglement inferred by the principle of maximum nonadditive entropy},\ }\href {https://doi.org/10.1103/PhysRevA.60.3461} {\bibfield  {journal} {\bibinfo  {journal} {Physical Review A}\ }\textbf {\bibinfo {volume} {60}},\ \bibinfo {pages} {3461} (\bibinfo {year} {1999})}\BibitemShut {NoStop}%
\bibitem [{\citenamefont {Ziman}(2008)}]{Zim08}%
  \BibitemOpen
  \bibfield  {author} {\bibinfo {author} {\bibfnamefont {M.}~\bibnamefont {Ziman}},\ }\bibfield  {title} {\bibinfo {title} {Incomplete quantum process tomography and principle of maximal entropy},\ }\href {https://doi.org/10.1103/PhysRevA.78.032118} {\bibfield  {journal} {\bibinfo  {journal} {Physical Review A}\ }\textbf {\bibinfo {volume} {78}},\ \bibinfo {pages} {032118} (\bibinfo {year} {2008})}\BibitemShut {NoStop}%
\bibitem [{\citenamefont {Gupta}\ \emph {et~al.}(2021)\citenamefont {Gupta}, \citenamefont {Xia}, \citenamefont {Levine},\ and\ \citenamefont {Kais}}]{GXLK21}%
  \BibitemOpen
  \bibfield  {author} {\bibinfo {author} {\bibfnamefont {R.}~\bibnamefont {Gupta}}, \bibinfo {author} {\bibfnamefont {R.}~\bibnamefont {Xia}}, \bibinfo {author} {\bibfnamefont {R.~D.}\ \bibnamefont {Levine}},\ and\ \bibinfo {author} {\bibfnamefont {S.}~\bibnamefont {Kais}},\ }\bibfield  {title} {\bibinfo {title} {Maximal entropy approach for quantum state tomography},\ }\href {https://doi.org/10.1103/PRXQuantum.2.010318} {\bibfield  {journal} {\bibinfo  {journal} {PRX Quantum}\ }\textbf {\bibinfo {volume} {2}},\ \bibinfo {pages} {010318} (\bibinfo {year} {2021})}\BibitemShut {NoStop}%
\bibitem [{\citenamefont {Bu}\ \emph {et~al.}(2023)\citenamefont {Bu}, \citenamefont {Gu},\ and\ \citenamefont {Jaffe}}]{BGJ23}%
  \BibitemOpen
  \bibfield  {author} {\bibinfo {author} {\bibfnamefont {K.}~\bibnamefont {Bu}}, \bibinfo {author} {\bibfnamefont {W.}~\bibnamefont {Gu}},\ and\ \bibinfo {author} {\bibfnamefont {A.}~\bibnamefont {Jaffe}},\ }\bibfield  {title} {\bibinfo {title} {Quantum entropy and central limit theorem},\ }\bibfield  {journal} {\bibinfo  {journal} {Proceedings of the National Academy of Sciences}\ }\textbf {\bibinfo {volume} {120}},\ \href {https://doi.org/10.1073/pnas.2304589120} {10.1073/pnas.2304589120} (\bibinfo {year} {2023})\BibitemShut {NoStop}%
\bibitem [{\citenamefont {Scarpa}\ \emph {et~al.}(2025)\citenamefont {Scarpa}, \citenamefont {Alhajri}, \citenamefont {Vedral},\ and\ \citenamefont {Anza}}]{SAVA25}%
  \BibitemOpen
  \bibfield  {author} {\bibinfo {author} {\bibfnamefont {L.}~\bibnamefont {Scarpa}}, \bibinfo {author} {\bibfnamefont {A.}~\bibnamefont {Alhajri}}, \bibinfo {author} {\bibfnamefont {V.}~\bibnamefont {Vedral}},\ and\ \bibinfo {author} {\bibfnamefont {F.}~\bibnamefont {Anza}},\ }\href {https://arxiv.org/abs/2309.15173} {\bibinfo {title} {Observable statistical mechanics}} (\bibinfo {year} {2025}),\ \Eprint {https://arxiv.org/abs/2309.15173} {arXiv:2309.15173 [quant-ph]} \BibitemShut {NoStop}%
\bibitem [{\citenamefont {Rigol}\ \emph {et~al.}(2008)\citenamefont {Rigol}, \citenamefont {Dunjko},\ and\ \citenamefont {Olshanii}}]{RDO08}%
  \BibitemOpen
  \bibfield  {author} {\bibinfo {author} {\bibfnamefont {M.}~\bibnamefont {Rigol}}, \bibinfo {author} {\bibfnamefont {V.}~\bibnamefont {Dunjko}},\ and\ \bibinfo {author} {\bibfnamefont {M.}~\bibnamefont {Olshanii}},\ }\bibfield  {title} {\bibinfo {title} {Thermalization and its mechanism for generic isolated quantum systems},\ }\href {https://doi.org/10.1038/nature06838} {\bibfield  {journal} {\bibinfo  {journal} {Nature}\ }\textbf {\bibinfo {volume} {452}},\ \bibinfo {pages} {854–858} (\bibinfo {year} {2008})}\BibitemShut {NoStop}%
\bibitem [{\citenamefont {Ba\~nuls}\ \emph {et~al.}(2011)\citenamefont {Ba\~nuls}, \citenamefont {Cirac},\ and\ \citenamefont {Hastings}}]{BCH11}%
  \BibitemOpen
  \bibfield  {author} {\bibinfo {author} {\bibfnamefont {M.~C.}\ \bibnamefont {Ba\~nuls}}, \bibinfo {author} {\bibfnamefont {J.~I.}\ \bibnamefont {Cirac}},\ and\ \bibinfo {author} {\bibfnamefont {M.~B.}\ \bibnamefont {Hastings}},\ }\bibfield  {title} {\bibinfo {title} {Strong and weak thermalization of infinite nonintegrable quantum systems},\ }\href {https://doi.org/10.1103/PhysRevLett.106.050405} {\bibfield  {journal} {\bibinfo  {journal} {Physical Review Letters}\ }\textbf {\bibinfo {volume} {106}},\ \bibinfo {pages} {050405} (\bibinfo {year} {2011})}\BibitemShut {NoStop}%
\bibitem [{\citenamefont {Doyon}(2017)}]{Ben17}%
  \BibitemOpen
  \bibfield  {author} {\bibinfo {author} {\bibfnamefont {B.}~\bibnamefont {Doyon}},\ }\bibfield  {title} {\bibinfo {title} {Thermalization and pseudolocality in extended quantum systems},\ }\href {https://doi.org/10.1007/s00220-017-2836-7} {\bibfield  {journal} {\bibinfo  {journal} {Communications in Mathematical Physics}\ }\textbf {\bibinfo {volume} {351}},\ \bibinfo {pages} {155} (\bibinfo {year} {2017})}\BibitemShut {NoStop}%
\bibitem [{\citenamefont {Mori}\ \emph {et~al.}(2018)\citenamefont {Mori}, \citenamefont {Ikeda}, \citenamefont {Kaminishi},\ and\ \citenamefont {Ueda}}]{MIKU18}%
  \BibitemOpen
  \bibfield  {author} {\bibinfo {author} {\bibfnamefont {T.}~\bibnamefont {Mori}}, \bibinfo {author} {\bibfnamefont {T.~N.}\ \bibnamefont {Ikeda}}, \bibinfo {author} {\bibfnamefont {E.}~\bibnamefont {Kaminishi}},\ and\ \bibinfo {author} {\bibfnamefont {M.}~\bibnamefont {Ueda}},\ }\bibfield  {title} {\bibinfo {title} {Thermalization and prethermalization in isolated quantum systems: a theoretical overview},\ }\href {https://doi.org/10.1088/1361-6455/aabcdf} {\bibfield  {journal} {\bibinfo  {journal} {Journal of Physics B: Atomic, Molecular and Optical Physics}\ }\textbf {\bibinfo {volume} {51}},\ \bibinfo {pages} {112001} (\bibinfo {year} {2018})}\BibitemShut {NoStop}%
\bibitem [{\citenamefont {Mandal}\ \emph {et~al.}(2018)\citenamefont {Mandal}, \citenamefont {Paranjape},\ and\ \citenamefont {Sorokhaibam}}]{MPS18}%
  \BibitemOpen
  \bibfield  {author} {\bibinfo {author} {\bibfnamefont {G.}~\bibnamefont {Mandal}}, \bibinfo {author} {\bibfnamefont {S.}~\bibnamefont {Paranjape}},\ and\ \bibinfo {author} {\bibfnamefont {N.}~\bibnamefont {Sorokhaibam}},\ }\bibfield  {title} {\bibinfo {title} {Thermalization in {2D} critical quench and uv/ir mixing},\ }\href {https://doi.org/10.1007/JHEP01(2018)027} {\bibfield  {journal} {\bibinfo  {journal} {Journal of High Energy Physics}\ }\textbf {\bibinfo {volume} {2018}},\ \bibinfo {pages} {27} (\bibinfo {year} {2018})}\BibitemShut {NoStop}%
\bibitem [{\citenamefont {Pathria}\ and\ \citenamefont {Beale}(2021)}]{PB21}%
  \BibitemOpen
  \bibfield  {author} {\bibinfo {author} {\bibfnamefont {R.~K.}\ \bibnamefont {Pathria}}\ and\ \bibinfo {author} {\bibfnamefont {P.~D.}\ \bibnamefont {Beale}},\ }\href {https://doi.org/10.1016/c2017-0-01713-5} {\emph {\bibinfo {title} {Statistical Mechanics: International Series of Monographs in Natural Philosophy}}}\ (\bibinfo  {publisher} {Elsevier},\ \bibinfo {year} {2021})\BibitemShut {NoStop}%
\bibitem [{\citenamefont {Devulapalli}\ \emph {et~al.}(2025)\citenamefont {Devulapalli}, \citenamefont {Mooney},\ and\ \citenamefont {Watson}}]{DMW25}%
  \BibitemOpen
  \bibfield  {author} {\bibinfo {author} {\bibfnamefont {D.}~\bibnamefont {Devulapalli}}, \bibinfo {author} {\bibfnamefont {T.~C.}\ \bibnamefont {Mooney}},\ and\ \bibinfo {author} {\bibfnamefont {J.~D.}\ \bibnamefont {Watson}},\ }\href {https://arxiv.org/abs/2507.00405} {\bibinfo {title} {The complexity of thermalization in finite quantum systems}} (\bibinfo {year} {2025}),\ \bibinfo {note} {arXiv:2507.00405}\BibitemShut {NoStop}%
\bibitem [{\citenamefont {Devetak}\ \emph {et~al.}(2006)\citenamefont {Devetak}, \citenamefont {Junge}, \citenamefont {King},\ and\ \citenamefont {Ruskai}}]{DJKR06}%
  \BibitemOpen
  \bibfield  {author} {\bibinfo {author} {\bibfnamefont {I.}~\bibnamefont {Devetak}}, \bibinfo {author} {\bibfnamefont {M.}~\bibnamefont {Junge}}, \bibinfo {author} {\bibfnamefont {C.}~\bibnamefont {King}},\ and\ \bibinfo {author} {\bibfnamefont {M.~B.}\ \bibnamefont {Ruskai}},\ }\bibfield  {title} {\bibinfo {title} {Multiplicativity of completely bounded p-norms implies a new additivity result},\ }\href {https://doi.org/10.1007/s00220-006-0034-0} {\bibfield  {journal} {\bibinfo  {journal} {Communications in Mathematical Physics}\ }\textbf {\bibinfo {volume} {266}},\ \bibinfo {pages} {37–63} (\bibinfo {year} {2006})}\BibitemShut {NoStop}%
\bibitem [{\citenamefont {Gour}\ and\ \citenamefont {Wilde}(2021)}]{GW21}%
  \BibitemOpen
  \bibfield  {author} {\bibinfo {author} {\bibfnamefont {G.}~\bibnamefont {Gour}}\ and\ \bibinfo {author} {\bibfnamefont {M.~M.}\ \bibnamefont {Wilde}},\ }\bibfield  {title} {\bibinfo {title} {Entropy of a quantum channel},\ }\bibfield  {journal} {\bibinfo  {journal} {Physical Review Research}\ }\textbf {\bibinfo {volume} {3}},\ \href {https://doi.org/10.1103/physrevresearch.3.023096} {10.1103/physrevresearch.3.023096} (\bibinfo {year} {2021})\BibitemShut {NoStop}%
\bibitem [{\citenamefont {Badhani}\ \emph {et~al.}(2026)\citenamefont {Badhani}, \citenamefont {GS}, \citenamefont {Choudhary}, \citenamefont {Anand},\ and\ \citenamefont {Das}}]{BSC+25}%
  \BibitemOpen
  \bibfield  {author} {\bibinfo {author} {\bibfnamefont {H.}~\bibnamefont {Badhani}}, \bibinfo {author} {\bibfnamefont {D.}~\bibnamefont {GS}}, \bibinfo {author} {\bibfnamefont {S.}~\bibnamefont {Choudhary}}, \bibinfo {author} {\bibfnamefont {V.}~\bibnamefont {Anand}},\ and\ \bibinfo {author} {\bibfnamefont {S.}~\bibnamefont {Das}},\ }\bibfield  {title} {\bibinfo {title} {Erasure cost of a quantum process: a thermodynamic meaning of the dynamical min-entropy},\ }\href {https://doi.org/10.1088/2058-9565/ae34e2} {\bibfield  {journal} {\bibinfo  {journal} {Quantum Science and Technology}\ }\textbf {\bibinfo {volume} {11}},\ \bibinfo {pages} {015038} (\bibinfo {year} {2026})},\ \bibinfo {note} {see also arXiv:2506.05307}\BibitemShut {NoStop}%
\bibitem [{\citenamefont {Faist}\ \emph {et~al.}(2021)\citenamefont {Faist}, \citenamefont {Berta},\ and\ \citenamefont {Brandao}}]{FBB21}%
  \BibitemOpen
  \bibfield  {author} {\bibinfo {author} {\bibfnamefont {P.}~\bibnamefont {Faist}}, \bibinfo {author} {\bibfnamefont {M.}~\bibnamefont {Berta}},\ and\ \bibinfo {author} {\bibfnamefont {F.~G. S.~L.}\ \bibnamefont {Brandao}},\ }\bibfield  {title} {\bibinfo {title} {Thermodynamic implementations of quantum processes},\ }\href {https://doi.org/10.1007/s00220-021-04107-w} {\bibfield  {journal} {\bibinfo  {journal} {Communications in Mathematical Physics}\ }\textbf {\bibinfo {volume} {384}},\ \bibinfo {pages} {1709–1750} (\bibinfo {year} {2021})}\BibitemShut {NoStop}%
\bibitem [{\citenamefont {Chiribella}\ \emph {et~al.}(2021)\citenamefont {Chiribella}, \citenamefont {Yang},\ and\ \citenamefont {Renner}}]{CYR21}%
  \BibitemOpen
  \bibfield  {author} {\bibinfo {author} {\bibfnamefont {G.}~\bibnamefont {Chiribella}}, \bibinfo {author} {\bibfnamefont {Y.}~\bibnamefont {Yang}},\ and\ \bibinfo {author} {\bibfnamefont {R.}~\bibnamefont {Renner}},\ }\bibfield  {title} {\bibinfo {title} {Fundamental energy requirement of reversible quantum operations},\ }\href {https://doi.org/10.1103/PhysRevX.11.021014} {\bibfield  {journal} {\bibinfo  {journal} {Physical Review X}\ }\textbf {\bibinfo {volume} {11}},\ \bibinfo {pages} {021014} (\bibinfo {year} {2021})}\BibitemShut {NoStop}%
\bibitem [{\citenamefont {Badhani}\ and\ \citenamefont {Das}(2026)}]{BD26}%
  \BibitemOpen
  \bibfield  {author} {\bibinfo {author} {\bibfnamefont {H.}~\bibnamefont {Badhani}}\ and\ \bibinfo {author} {\bibfnamefont {S.}~\bibnamefont {Das}},\ }\href {https://arxiv.org/abs/2604.01217} {\bibinfo {title} {Conditional channel entropy sets fundamental limits on thermodynamic quantum information processing}} (\bibinfo {year} {2026}),\ \Eprint {https://arxiv.org/abs/2604.01217} {arXiv:2604.01217 [quant-ph]} \BibitemShut {NoStop}%
\bibitem [{\citenamefont {Kraus}(1983)}]{Kra83}%
  \BibitemOpen
  \bibfield  {author} {\bibinfo {author} {\bibfnamefont {K.}~\bibnamefont {Kraus}},\ }\href {https://books.google.co.in/books?id=fRBBAQAAIAAJ} {\emph {\bibinfo {title} {States, effects, and operations}}}\ (\bibinfo  {publisher} {Springer, Berlin},\ \bibinfo {year} {1983})\BibitemShut {NoStop}%
\bibitem [{\citenamefont {Sudarshan}(1985)}]{Sud85}%
  \BibitemOpen
  \bibfield  {author} {\bibinfo {author} {\bibfnamefont {E.~C.~G.}\ \bibnamefont {Sudarshan}},\ }\href {{http://inis.iaea.org/search/search.aspx?orig_q=RN:17063283}} {\emph {\bibinfo {title} {Quantum measurement and dynamical maps}}}\ (\bibinfo  {publisher} {Cambridge University Press},\ \bibinfo {address} {Cambridge},\ \bibinfo {year} {1985})\BibitemShut {NoStop}%
\bibitem [{\citenamefont {Stratton}\ \emph {et~al.}(2024)\citenamefont {Stratton}, \citenamefont {Hsieh},\ and\ \citenamefont {Skrzypczyk}}]{SHS24}%
  \BibitemOpen
  \bibfield  {author} {\bibinfo {author} {\bibfnamefont {B.}~\bibnamefont {Stratton}}, \bibinfo {author} {\bibfnamefont {C.-Y.}\ \bibnamefont {Hsieh}},\ and\ \bibinfo {author} {\bibfnamefont {P.}~\bibnamefont {Skrzypczyk}},\ }\bibfield  {title} {\bibinfo {title} {Dynamical resource theory of informational nonequilibrium preservability},\ }\href {https://doi.org/10.1103/PhysRevLett.132.110202} {\bibfield  {journal} {\bibinfo  {journal} {Physical Review Letters}\ }\textbf {\bibinfo {volume} {132}},\ \bibinfo {pages} {110202} (\bibinfo {year} {2024})}\BibitemShut {NoStop}%
\bibitem [{\citenamefont {Brandão}\ \emph {et~al.}(2015)\citenamefont {Brandão}, \citenamefont {Horodecki}, \citenamefont {Ng}, \citenamefont {Oppenheim},\ and\ \citenamefont {Wehner}}]{BHN+15}%
  \BibitemOpen
  \bibfield  {author} {\bibinfo {author} {\bibfnamefont {F.}~\bibnamefont {Brandão}}, \bibinfo {author} {\bibfnamefont {M.}~\bibnamefont {Horodecki}}, \bibinfo {author} {\bibfnamefont {N.}~\bibnamefont {Ng}}, \bibinfo {author} {\bibfnamefont {J.}~\bibnamefont {Oppenheim}},\ and\ \bibinfo {author} {\bibfnamefont {S.}~\bibnamefont {Wehner}},\ }\bibfield  {title} {\bibinfo {title} {The second laws of quantum thermodynamics},\ }\href {https://doi.org/10.1073/pnas.1411728112} {\bibfield  {journal} {\bibinfo  {journal} {Proceedings of the National Academy of Sciences}\ }\textbf {\bibinfo {volume} {112}},\ \bibinfo {pages} {3275–3279} (\bibinfo {year} {2015})}\BibitemShut {NoStop}%
\bibitem [{\citenamefont {Sparaciari}\ \emph {et~al.}(2020)\citenamefont {Sparaciari}, \citenamefont {del Rio}, \citenamefont {Scandolo}, \citenamefont {Faist},\ and\ \citenamefont {Oppenheim}}]{SdS+20}%
  \BibitemOpen
  \bibfield  {author} {\bibinfo {author} {\bibfnamefont {C.}~\bibnamefont {Sparaciari}}, \bibinfo {author} {\bibfnamefont {L.}~\bibnamefont {del Rio}}, \bibinfo {author} {\bibfnamefont {C.~M.}\ \bibnamefont {Scandolo}}, \bibinfo {author} {\bibfnamefont {P.}~\bibnamefont {Faist}},\ and\ \bibinfo {author} {\bibfnamefont {J.}~\bibnamefont {Oppenheim}},\ }\bibfield  {title} {\bibinfo {title} {The first law of general quantum resource theories},\ }\href {https://doi.org/10.22331/q-2020-04-30-259} {\bibfield  {journal} {\bibinfo  {journal} {Quantum}\ }\textbf {\bibinfo {volume} {4}},\ \bibinfo {pages} {259} (\bibinfo {year} {2020})}\BibitemShut {NoStop}%
\bibitem [{\citenamefont {Badhani}\ \emph {et~al.}(2025{\natexlab{a}})\citenamefont {Badhani}, \citenamefont {S},\ and\ \citenamefont {Das}}]{BGDb}%
  \BibitemOpen
  \bibfield  {author} {\bibinfo {author} {\bibfnamefont {H.}~\bibnamefont {Badhani}}, \bibinfo {author} {\bibfnamefont {D.~G.}\ \bibnamefont {S}},\ and\ \bibinfo {author} {\bibfnamefont {S.}~\bibnamefont {Das}},\ }\href {https://arxiv.org/abs/2510.23731} {\bibinfo {title} {Thermodynamic work capacity of quantum information processing}} (\bibinfo {year} {2025}{\natexlab{a}}),\ \Eprint {https://arxiv.org/abs/2510.23731} {arXiv:2510.23731 [quant-ph]} \BibitemShut {NoStop}%
\bibitem [{\citenamefont {Sohail}\ \emph {et~al.}(2026)\citenamefont {Sohail}, \citenamefont {Pandey}, \citenamefont {Singh},\ and\ \citenamefont {Das}}]{SPSD25}%
  \BibitemOpen
  \bibfield  {author} {\bibinfo {author} {\bibnamefont {Sohail}}, \bibinfo {author} {\bibfnamefont {V.}~\bibnamefont {Pandey}}, \bibinfo {author} {\bibfnamefont {U.}~\bibnamefont {Singh}},\ and\ \bibinfo {author} {\bibfnamefont {S.}~\bibnamefont {Das}},\ }\bibfield  {title} {\bibinfo {title} {Fundamental limitations on the recoverability of quantum processes},\ }\href {https://doi.org/10.1007/s00023-025-01590-y} {\bibfield  {journal} {\bibinfo  {journal} {Annales Henri Poincar{\'e}}\ }\textbf {\bibinfo {volume} {27}},\ \bibinfo {pages} {933} (\bibinfo {year} {2026})}\BibitemShut {NoStop}%
\bibitem [{\citenamefont {Yang}\ \emph {et~al.}(2019)\citenamefont {Yang}, \citenamefont {Horodecki},\ and\ \citenamefont {Winter}}]{YHW19}%
  \BibitemOpen
  \bibfield  {author} {\bibinfo {author} {\bibfnamefont {D.}~\bibnamefont {Yang}}, \bibinfo {author} {\bibfnamefont {K.}~\bibnamefont {Horodecki}},\ and\ \bibinfo {author} {\bibfnamefont {A.}~\bibnamefont {Winter}},\ }\bibfield  {title} {\bibinfo {title} {Distributed private randomness distillation},\ }\href {https://doi.org/10.1103/PhysRevLett.123.170501} {\bibfield  {journal} {\bibinfo  {journal} {Physical Review Letters}\ }\textbf {\bibinfo {volume} {123}},\ \bibinfo {pages} {170501} (\bibinfo {year} {2019})}\BibitemShut {NoStop}%
\bibitem [{\citenamefont {Cerf}\ and\ \citenamefont {Adami}(1997)}]{CA97}%
  \BibitemOpen
  \bibfield  {author} {\bibinfo {author} {\bibfnamefont {N.~J.}\ \bibnamefont {Cerf}}\ and\ \bibinfo {author} {\bibfnamefont {C.}~\bibnamefont {Adami}},\ }\bibfield  {title} {\bibinfo {title} {Negative entropy and information in quantum mechanics},\ }\href {https://doi.org/10.1103/PhysRevLett.79.5194} {\bibfield  {journal} {\bibinfo  {journal} {Physical Review Letters}\ }\textbf {\bibinfo {volume} {79}},\ \bibinfo {pages} {5194} (\bibinfo {year} {1997})}\BibitemShut {NoStop}%
\bibitem [{\citenamefont {Shirokov}\ and\ \citenamefont {Bulinski}(2020)}]{SB20}%
  \BibitemOpen
  \bibfield  {author} {\bibinfo {author} {\bibfnamefont {M.~E.}\ \bibnamefont {Shirokov}}\ and\ \bibinfo {author} {\bibfnamefont {A.~V.}\ \bibnamefont {Bulinski}},\ }\bibfield  {title} {\bibinfo {title} {On quantum channels and operations preserving finiteness of the von {N}eumann entropy},\ }\href {https://doi.org/10.1134/s1995080220120392} {\bibfield  {journal} {\bibinfo  {journal} {Lobachevskii Journal of Mathematics}\ }\textbf {\bibinfo {volume} {41}},\ \bibinfo {pages} {2383–2396} (\bibinfo {year} {2020})}\BibitemShut {NoStop}%
\bibitem [{\citenamefont {Lenard}(1978)}]{Len78}%
  \BibitemOpen
  \bibfield  {author} {\bibinfo {author} {\bibfnamefont {A.}~\bibnamefont {Lenard}},\ }\bibfield  {title} {\bibinfo {title} {Thermodynamical proof of the {G}ibbs formula for elementary quantum systems},\ }\href {https://doi.org/10.1007/BF01011769} {\bibfield  {journal} {\bibinfo  {journal} {Journal of Statistical Physics}\ }\textbf {\bibinfo {volume} {19}},\ \bibinfo {pages} {575} (\bibinfo {year} {1978})}\BibitemShut {NoStop}%
\bibitem [{\citenamefont {Cooney}\ \emph {et~al.}(2016)\citenamefont {Cooney}, \citenamefont {Mosonyi},\ and\ \citenamefont {Wilde}}]{CMW16}%
  \BibitemOpen
  \bibfield  {author} {\bibinfo {author} {\bibfnamefont {T.}~\bibnamefont {Cooney}}, \bibinfo {author} {\bibfnamefont {M.}~\bibnamefont {Mosonyi}},\ and\ \bibinfo {author} {\bibfnamefont {M.~M.}\ \bibnamefont {Wilde}},\ }\bibfield  {title} {\bibinfo {title} {Strong converse exponents for a quantum channel discrimination problem and quantum‐feedback‐assisted communication},\ }\href {https://doi.org/10.1007/s00220-016-2645-4} {\bibfield  {journal} {\bibinfo  {journal} {Communications in Mathematical Physics}\ }\textbf {\bibinfo {volume} {344}},\ \bibinfo {pages} {797} (\bibinfo {year} {2016})},\ \bibinfo {note} {arXiv:1408.3373}\BibitemShut {NoStop}%
\bibitem [{\citenamefont {Das}\ and\ \citenamefont {Wilde}(2019)}]{DW19}%
  \BibitemOpen
  \bibfield  {author} {\bibinfo {author} {\bibfnamefont {S.}~\bibnamefont {Das}}\ and\ \bibinfo {author} {\bibfnamefont {M.~M.}\ \bibnamefont {Wilde}},\ }\bibfield  {title} {\bibinfo {title} {Quantum rebound capacity},\ }\bibfield  {journal} {\bibinfo  {journal} {Physical Review A}\ }\textbf {\bibinfo {volume} {100}},\ \href {https://doi.org/10.1103/physreva.100.030302} {10.1103/physreva.100.030302} (\bibinfo {year} {2019})\BibitemShut {NoStop}%
\bibitem [{\citenamefont {Yuan}(2019)}]{Yua19}%
  \BibitemOpen
  \bibfield  {author} {\bibinfo {author} {\bibfnamefont {X.}~\bibnamefont {Yuan}},\ }\bibfield  {title} {\bibinfo {title} {Hypothesis testing and entropies of quantum channels},\ }\href {https://doi.org/10.1103/PhysRevA.99.032317} {\bibfield  {journal} {\bibinfo  {journal} {Physical Review A}\ }\textbf {\bibinfo {volume} {99}},\ \bibinfo {pages} {032317} (\bibinfo {year} {2019})}\BibitemShut {NoStop}%
\bibitem [{\citenamefont {Giovannetti}\ \emph {et~al.}(2004)\citenamefont {Giovannetti}, \citenamefont {Guha}, \citenamefont {Lloyd}, \citenamefont {Maccone},\ and\ \citenamefont {Shapiro}}]{GGL+04}%
  \BibitemOpen
  \bibfield  {author} {\bibinfo {author} {\bibfnamefont {V.}~\bibnamefont {Giovannetti}}, \bibinfo {author} {\bibfnamefont {S.}~\bibnamefont {Guha}}, \bibinfo {author} {\bibfnamefont {S.}~\bibnamefont {Lloyd}}, \bibinfo {author} {\bibfnamefont {L.}~\bibnamefont {Maccone}},\ and\ \bibinfo {author} {\bibfnamefont {J.~H.}\ \bibnamefont {Shapiro}},\ }\bibfield  {title} {\bibinfo {title} {Minimum output entropy of bosonic channels: A conjecture},\ }\href {https://doi.org/10.1103/PhysRevA.70.032315} {\bibfield  {journal} {\bibinfo  {journal} {Physical Review A}\ }\textbf {\bibinfo {volume} {70}},\ \bibinfo {pages} {032315} (\bibinfo {year} {2004})}\BibitemShut {NoStop}%
\bibitem [{\citenamefont {De~Palma}\ \emph {et~al.}(2017)\citenamefont {De~Palma}, \citenamefont {Trevisan},\ and\ \citenamefont {Giovannetti}}]{DTG17}%
  \BibitemOpen
  \bibfield  {author} {\bibinfo {author} {\bibfnamefont {G.}~\bibnamefont {De~Palma}}, \bibinfo {author} {\bibfnamefont {D.}~\bibnamefont {Trevisan}},\ and\ \bibinfo {author} {\bibfnamefont {V.}~\bibnamefont {Giovannetti}},\ }\bibfield  {title} {\bibinfo {title} {Gaussian states minimize the output entropy of one-mode quantum {G}aussian channels},\ }\href {https://doi.org/10.1103/PhysRevLett.118.160503} {\bibfield  {journal} {\bibinfo  {journal} {Physical Review Letters}\ }\textbf {\bibinfo {volume} {118}},\ \bibinfo {pages} {160503} (\bibinfo {year} {2017})}\BibitemShut {NoStop}%
\bibitem [{\citenamefont {Liu}\ and\ \citenamefont {Yuan}(2020)}]{LY20}%
  \BibitemOpen
  \bibfield  {author} {\bibinfo {author} {\bibfnamefont {Y.}~\bibnamefont {Liu}}\ and\ \bibinfo {author} {\bibfnamefont {X.}~\bibnamefont {Yuan}},\ }\bibfield  {title} {\bibinfo {title} {Operational resource theory of quantum channels},\ }\bibfield  {journal} {\bibinfo  {journal} {Physical Review Research}\ }\textbf {\bibinfo {volume} {2}},\ \href {https://doi.org/10.1103/physrevresearch.2.012035} {10.1103/physrevresearch.2.012035} (\bibinfo {year} {2020})\BibitemShut {NoStop}%
\bibitem [{\citenamefont {Badhani}\ \emph {et~al.}(2025{\natexlab{b}})\citenamefont {Badhani}, \citenamefont {S},\ and\ \citenamefont {Das}}]{BGD25}%
  \BibitemOpen
  \bibfield  {author} {\bibinfo {author} {\bibfnamefont {H.}~\bibnamefont {Badhani}}, \bibinfo {author} {\bibfnamefont {D.~G.}\ \bibnamefont {S}},\ and\ \bibinfo {author} {\bibfnamefont {S.}~\bibnamefont {Das}},\ }\href {https://arxiv.org/abs/2510.12790} {\bibinfo {title} {Thermodynamics of quantum processes: An operational framework for free energy and reversible athermality}} (\bibinfo {year} {2025}{\natexlab{b}}),\ \Eprint {https://arxiv.org/abs/2510.12790} {arXiv:2510.12790 [quant-ph]} \BibitemShut {NoStop}%
\bibitem [{\citenamefont {Winter}(2017)}]{Win17}%
  \BibitemOpen
  \bibfield  {author} {\bibinfo {author} {\bibfnamefont {A.}~\bibnamefont {Winter}},\ }\href {https://arxiv.org/abs/1712.10267} {\bibinfo {title} {Energy-constrained diamond norm with applications to the uniform continuity of continuous variable channel capacities}} (\bibinfo {year} {2017}),\ \Eprint {https://arxiv.org/abs/1712.10267} {arXiv:1712.10267 [quant-ph]} \BibitemShut {NoStop}%
\bibitem [{\citenamefont {van Luijk}(2025)}]{vL25}%
  \BibitemOpen
  \bibfield  {author} {\bibinfo {author} {\bibfnamefont {L.}~\bibnamefont {van Luijk}},\ }\bibfield  {title} {\bibinfo {title} {Energy-limited quantum dynamics},\ }\bibfield  {journal} {\bibinfo  {journal} {Communications in Mathematical Physics}\ }\textbf {\bibinfo {volume} {406}},\ \href {https://doi.org/10.1007/s00220-025-05282-w} {10.1007/s00220-025-05282-w} (\bibinfo {year} {2025})\BibitemShut {NoStop}%
\bibitem [{\citenamefont {Haag}\ \emph {et~al.}(1967)\citenamefont {Haag}, \citenamefont {Hugenholtz},\ and\ \citenamefont {Winnink}}]{HHW67}%
  \BibitemOpen
  \bibfield  {author} {\bibinfo {author} {\bibfnamefont {R.}~\bibnamefont {Haag}}, \bibinfo {author} {\bibfnamefont {N.~M.}\ \bibnamefont {Hugenholtz}},\ and\ \bibinfo {author} {\bibfnamefont {M.}~\bibnamefont {Winnink}},\ }\bibfield  {title} {\bibinfo {title} {On the equilibrium states in quantum statistical mechanics},\ }\href {https://doi.org/10.1007/BF01646342} {\bibfield  {journal} {\bibinfo  {journal} {Communications in Mathematical Physics}\ }\textbf {\bibinfo {volume} {5}},\ \bibinfo {pages} {215} (\bibinfo {year} {1967})}\BibitemShut {NoStop}%
\bibitem [{\citenamefont {Scarani}\ \emph {et~al.}(2002)\citenamefont {Scarani}, \citenamefont {Ziman}, \citenamefont {\ifmmode \check{S}\else \v{S}\fi{}telmachovi\ifmmode~\check{c}\else \v{c}\fi{}}, \citenamefont {Gisin},\ and\ \citenamefont {Bu\ifmmode~\check{z}\else \v{z}\fi{}ek}}]{SZS+05}%
  \BibitemOpen
  \bibfield  {author} {\bibinfo {author} {\bibfnamefont {V.}~\bibnamefont {Scarani}}, \bibinfo {author} {\bibfnamefont {M.}~\bibnamefont {Ziman}}, \bibinfo {author} {\bibfnamefont {P.}~\bibnamefont {\ifmmode \check{S}\else \v{S}\fi{}telmachovi\ifmmode~\check{c}\else \v{c}\fi{}}}, \bibinfo {author} {\bibfnamefont {N.}~\bibnamefont {Gisin}},\ and\ \bibinfo {author} {\bibfnamefont {V.}~\bibnamefont {Bu\ifmmode~\check{z}\else \v{z}\fi{}ek}},\ }\bibfield  {title} {\bibinfo {title} {Thermalizing quantum machines: Dissipation and entanglement},\ }\href {https://doi.org/10.1103/PhysRevLett.88.097905} {\bibfield  {journal} {\bibinfo  {journal} {Physical Review Letters}\ }\textbf {\bibinfo {volume} {88}},\ \bibinfo {pages} {097905} (\bibinfo {year} {2002})}\BibitemShut {NoStop}%
\bibitem [{\citenamefont {Hallam}\ \emph {et~al.}(2019)\citenamefont {Hallam}, \citenamefont {Morley},\ and\ \citenamefont {Green}}]{HMG19}%
  \BibitemOpen
  \bibfield  {author} {\bibinfo {author} {\bibfnamefont {A.}~\bibnamefont {Hallam}}, \bibinfo {author} {\bibfnamefont {J.~G.}\ \bibnamefont {Morley}},\ and\ \bibinfo {author} {\bibfnamefont {A.~G.}\ \bibnamefont {Green}},\ }\bibfield  {title} {\bibinfo {title} {The {L}yapunov spectra of quantum thermalisation},\ }\href {https://doi.org/10.1038/s41467-019-10336-4} {\bibfield  {journal} {\bibinfo  {journal} {Nature Communications}\ }\textbf {\bibinfo {volume} {10}},\ \bibinfo {pages} {2708} (\bibinfo {year} {2019})}\BibitemShut {NoStop}%
\bibitem [{\citenamefont {Mohan}\ \emph {et~al.}(2022)\citenamefont {Mohan}, \citenamefont {Das},\ and\ \citenamefont {Pati}}]{MDP22}%
  \BibitemOpen
  \bibfield  {author} {\bibinfo {author} {\bibfnamefont {B.}~\bibnamefont {Mohan}}, \bibinfo {author} {\bibfnamefont {S.}~\bibnamefont {Das}},\ and\ \bibinfo {author} {\bibfnamefont {A.~K.}\ \bibnamefont {Pati}},\ }\bibfield  {title} {\bibinfo {title} {Quantum speed limits for information and coherence},\ }\href {https://doi.org/10.1088/1367-2630/ac753c} {\bibfield  {journal} {\bibinfo  {journal} {New Journal of Physics}\ }\textbf {\bibinfo {volume} {24}},\ \bibinfo {pages} {065003} (\bibinfo {year} {2022})}\BibitemShut {NoStop}%
\bibitem [{\citenamefont {Zurek}(2001)}]{Zur01}%
  \BibitemOpen
  \bibfield  {author} {\bibinfo {author} {\bibfnamefont {W.~H.}\ \bibnamefont {Zurek}},\ }\href {https://arxiv.org/abs/quant-ph/0111137} {\bibinfo {title} {Information transfer in quantum measurements: irreversibility and amplification}} (\bibinfo {year} {2001}),\ \bibinfo {note} {arXiv:quant-ph/0111137}\BibitemShut {NoStop}%
\bibitem [{\citenamefont {Horodecki}\ \emph {et~al.}(2003{\natexlab{a}})\citenamefont {Horodecki}, \citenamefont {Horodecki},\ and\ \citenamefont {Oppenheim}}]{HHO03}%
  \BibitemOpen
  \bibfield  {author} {\bibinfo {author} {\bibfnamefont {M.}~\bibnamefont {Horodecki}}, \bibinfo {author} {\bibfnamefont {P.}~\bibnamefont {Horodecki}},\ and\ \bibinfo {author} {\bibfnamefont {J.}~\bibnamefont {Oppenheim}},\ }\bibfield  {title} {\bibinfo {title} {Reversible transformations from pure to mixed states and the unique measure of information},\ }\href {https://doi.org/10.1103/PhysRevA.67.062104} {\bibfield  {journal} {\bibinfo  {journal} {Physical Review A}\ }\textbf {\bibinfo {volume} {67}},\ \bibinfo {pages} {062104} (\bibinfo {year} {2003}{\natexlab{a}})}\BibitemShut {NoStop}%
\bibitem [{\citenamefont {Horodecki}\ \emph {et~al.}(2003{\natexlab{b}})\citenamefont {Horodecki}, \citenamefont {Horodecki}, \citenamefont {Horodecki}, \citenamefont {Horodecki}, \citenamefont {Oppenheim}, \citenamefont {Sen(De)},\ and\ \citenamefont {Sen}}]{HHH+03}%
  \BibitemOpen
  \bibfield  {author} {\bibinfo {author} {\bibfnamefont {M.}~\bibnamefont {Horodecki}}, \bibinfo {author} {\bibfnamefont {K.}~\bibnamefont {Horodecki}}, \bibinfo {author} {\bibfnamefont {P.}~\bibnamefont {Horodecki}}, \bibinfo {author} {\bibfnamefont {R.}~\bibnamefont {Horodecki}}, \bibinfo {author} {\bibfnamefont {J.}~\bibnamefont {Oppenheim}}, \bibinfo {author} {\bibfnamefont {A.}~\bibnamefont {Sen(De)}},\ and\ \bibinfo {author} {\bibfnamefont {U.}~\bibnamefont {Sen}},\ }\bibfield  {title} {\bibinfo {title} {Local information as a resource in distributed quantum systems},\ }\href {https://doi.org/10.1103/PhysRevLett.90.100402} {\bibfield  {journal} {\bibinfo  {journal} {Physical Review Letters}\ }\textbf {\bibinfo {volume} {90}},\ \bibinfo {pages} {100402} (\bibinfo {year} {2003}{\natexlab{b}})}\BibitemShut {NoStop}%
\bibitem [{\citenamefont {Schindler}\ \emph {et~al.}(2025)\citenamefont {Schindler}, \citenamefont {Strasberg}, \citenamefont {Galke}, \citenamefont {Winter},\ and\ \citenamefont {Jabbour}}]{SSG+25}%
  \BibitemOpen
  \bibfield  {author} {\bibinfo {author} {\bibfnamefont {J.}~\bibnamefont {Schindler}}, \bibinfo {author} {\bibfnamefont {P.}~\bibnamefont {Strasberg}}, \bibinfo {author} {\bibfnamefont {N.}~\bibnamefont {Galke}}, \bibinfo {author} {\bibfnamefont {A.}~\bibnamefont {Winter}},\ and\ \bibinfo {author} {\bibfnamefont {M.~G.}\ \bibnamefont {Jabbour}},\ }\href {https://arxiv.org/abs/2503.15612} {\bibinfo {title} {Unification of observational entropy with maximum entropy principles}} (\bibinfo {year} {2025}),\ \Eprint {https://arxiv.org/abs/2503.15612} {arXiv:2503.15612 [quant-ph]} \BibitemShut {NoStop}%
\bibitem [{\citenamefont {Oppenheim}\ \emph {et~al.}(2002)\citenamefont {Oppenheim}, \citenamefont {Horodecki}, \citenamefont {Horodecki},\ and\ \citenamefont {Horodecki}}]{OHHH02}%
  \BibitemOpen
  \bibfield  {author} {\bibinfo {author} {\bibfnamefont {J.}~\bibnamefont {Oppenheim}}, \bibinfo {author} {\bibfnamefont {M.}~\bibnamefont {Horodecki}}, \bibinfo {author} {\bibfnamefont {P.}~\bibnamefont {Horodecki}},\ and\ \bibinfo {author} {\bibfnamefont {R.}~\bibnamefont {Horodecki}},\ }\bibfield  {title} {\bibinfo {title} {Thermodynamical approach to quantifying quantum correlations},\ }\href {https://doi.org/10.1103/PhysRevLett.89.180402} {\bibfield  {journal} {\bibinfo  {journal} {Physical Review Letters}\ }\textbf {\bibinfo {volume} {89}},\ \bibinfo {pages} {180402} (\bibinfo {year} {2002})}\BibitemShut {NoStop}%
\bibitem [{\citenamefont {Horodecki}\ \emph {et~al.}(2005{\natexlab{b}})\citenamefont {Horodecki}, \citenamefont {Horodecki}, \citenamefont {Horodecki}, \citenamefont {Oppenheim}, \citenamefont {Sen(De)}, \citenamefont {Sen},\ and\ \citenamefont {Synak-Radtke}}]{HHH+05}%
  \BibitemOpen
  \bibfield  {author} {\bibinfo {author} {\bibfnamefont {M.}~\bibnamefont {Horodecki}}, \bibinfo {author} {\bibfnamefont {P.}~\bibnamefont {Horodecki}}, \bibinfo {author} {\bibfnamefont {R.}~\bibnamefont {Horodecki}}, \bibinfo {author} {\bibfnamefont {J.}~\bibnamefont {Oppenheim}}, \bibinfo {author} {\bibfnamefont {A.}~\bibnamefont {Sen(De)}}, \bibinfo {author} {\bibfnamefont {U.}~\bibnamefont {Sen}},\ and\ \bibinfo {author} {\bibfnamefont {B.}~\bibnamefont {Synak-Radtke}},\ }\bibfield  {title} {\bibinfo {title} {Local versus nonlocal information in quantum-information theory: Formalism and phenomena},\ }\href {https://doi.org/10.1103/PhysRevA.71.062307} {\bibfield  {journal} {\bibinfo  {journal} {Physical Review A}\ }\textbf {\bibinfo {volume} {71}},\ \bibinfo {pages} {062307} (\bibinfo {year} {2005}{\natexlab{b}})}\BibitemShut {NoStop}%
\bibitem [{\citenamefont {Meng}\ \emph {et~al.}(2024)\citenamefont {Meng}, \citenamefont {Curran}, \citenamefont {Senno}, \citenamefont {Wright}, \citenamefont {Farkas}, \citenamefont {Scarani},\ and\ \citenamefont {Ac\'{\i}n}}]{MCS+24}%
  \BibitemOpen
  \bibfield  {author} {\bibinfo {author} {\bibfnamefont {S.}~\bibnamefont {Meng}}, \bibinfo {author} {\bibfnamefont {F.}~\bibnamefont {Curran}}, \bibinfo {author} {\bibfnamefont {G.}~\bibnamefont {Senno}}, \bibinfo {author} {\bibfnamefont {V.~J.}\ \bibnamefont {Wright}}, \bibinfo {author} {\bibfnamefont {M.}~\bibnamefont {Farkas}}, \bibinfo {author} {\bibfnamefont {V.}~\bibnamefont {Scarani}},\ and\ \bibinfo {author} {\bibfnamefont {A.}~\bibnamefont {Ac\'{\i}n}},\ }\bibfield  {title} {\bibinfo {title} {Maximal intrinsic randomness of a quantum state},\ }\href {https://doi.org/10.1103/PhysRevA.110.L010403} {\bibfield  {journal} {\bibinfo  {journal} {Physical Review A}\ }\textbf {\bibinfo {volume} {110}},\ \bibinfo {pages} {L010403} (\bibinfo {year} {2024})}\BibitemShut {NoStop}%
\bibitem [{\citenamefont {Riera}\ \emph {et~al.}(2012)\citenamefont {Riera}, \citenamefont {Gogolin},\ and\ \citenamefont {Eisert}}]{RGE12}%
  \BibitemOpen
  \bibfield  {author} {\bibinfo {author} {\bibfnamefont {A.}~\bibnamefont {Riera}}, \bibinfo {author} {\bibfnamefont {C.}~\bibnamefont {Gogolin}},\ and\ \bibinfo {author} {\bibfnamefont {J.}~\bibnamefont {Eisert}},\ }\bibfield  {title} {\bibinfo {title} {Thermalization in nature and on a quantum computer},\ }\href {https://doi.org/10.1103/PhysRevLett.108.080402} {\bibfield  {journal} {\bibinfo  {journal} {Physical Review Letters}\ }\textbf {\bibinfo {volume} {108}},\ \bibinfo {pages} {080402} (\bibinfo {year} {2012})}\BibitemShut {NoStop}%
\bibitem [{\citenamefont {Leverrier}\ and\ \citenamefont {Cerf}(2009)}]{LC09}%
  \BibitemOpen
  \bibfield  {author} {\bibinfo {author} {\bibfnamefont {A.}~\bibnamefont {Leverrier}}\ and\ \bibinfo {author} {\bibfnamefont {N.~J.}\ \bibnamefont {Cerf}},\ }\bibfield  {title} {\bibinfo {title} {Quantum de {F}inetti theorem in phase-space representation},\ }\href {https://doi.org/10.1103/PhysRevA.80.010102} {\bibfield  {journal} {\bibinfo  {journal} {Physical Review A}\ }\textbf {\bibinfo {volume} {80}},\ \bibinfo {pages} {010102} (\bibinfo {year} {2009})}\BibitemShut {NoStop}%
\bibitem [{\citenamefont {Le}\ \emph {et~al.}(2021)\citenamefont {Le}, \citenamefont {Winter},\ and\ \citenamefont {Adesso}}]{LWG21}%
  \BibitemOpen
  \bibfield  {author} {\bibinfo {author} {\bibfnamefont {T.~P.}\ \bibnamefont {Le}}, \bibinfo {author} {\bibfnamefont {A.}~\bibnamefont {Winter}},\ and\ \bibinfo {author} {\bibfnamefont {G.}~\bibnamefont {Adesso}},\ }\bibfield  {title} {\bibinfo {title} {Thermality versus objectivity: Can they peacefully coexist?},\ }\href {https://doi.org/10.3390/e23111506} {\bibfield  {journal} {\bibinfo  {journal} {Entropy}\ }\textbf {\bibinfo {volume} {23}},\ \bibinfo {pages} {1506} (\bibinfo {year} {2021})}\BibitemShut {NoStop}%
\bibitem [{\citenamefont {Singh}\ \emph {et~al.}(2024)\citenamefont {Singh}, \citenamefont {Sawicki},\ and\ \citenamefont {Korbicz}}]{SSK24}%
  \BibitemOpen
  \bibfield  {author} {\bibinfo {author} {\bibfnamefont {U.}~\bibnamefont {Singh}}, \bibinfo {author} {\bibfnamefont {A.}~\bibnamefont {Sawicki}},\ and\ \bibinfo {author} {\bibfnamefont {J.~K.}\ \bibnamefont {Korbicz}},\ }\bibfield  {title} {\bibinfo {title} {Pointer states in the {Born-Markov} approximation},\ }\href {https://doi.org/10.1103/PhysRevLett.132.030203} {\bibfield  {journal} {\bibinfo  {journal} {Physical Review Letters}\ }\textbf {\bibinfo {volume} {132}},\ \bibinfo {pages} {030203} (\bibinfo {year} {2024})}\BibitemShut {NoStop}%
\end{thebibliography}%
\end{document}